\documentclass{article}

\usepackage{arxiv}

\usepackage[utf8]{inputenc} 
\usepackage[T1]{fontenc}    
\usepackage{hyperref}       
\usepackage{url}            
\usepackage{booktabs}       
\usepackage{amsfonts}       
\usepackage{nicefrac}       
\usepackage{microtype}      

\usepackage{amssymb}

\usepackage{lineno}

\usepackage{latexsym}

\usepackage[english]{babel}
\usepackage{enumerate}

\usepackage{array}

\usepackage{graphicx}
\usepackage{caption}
\usepackage{subcaption}

\usepackage{float}
\usepackage{amsthm}

\usepackage{amssymb}

\usepackage{dsfont}
\usepackage{mathtools}
\usepackage{tikz-cd}

\usepackage{graphicx}

\usepackage{enumitem}

\makeatletter
\newcommand\mathcircled[1]{%
  \mathpalette\@mathcircled{#1}%
}
\newcommand\@mathcircled[2]{%
  \tikz[baseline=(math.base)] \node[draw,circle,inner sep=1pt] (math) {$\m@th#1#2$};%
}
\makeatother

\usepackage{hyperref}
    \hypersetup{
        colorlinks   = true,
        citecolor    = blue,
        linkcolor=blue  
    }

\newtheorem{theo}{\textbf{Theorem}}[section]
\newtheorem{lem}[theo]{\textbf{Lemma}}

\title{Analysis of a functional response with prey-density dependent handling time from an evolutionary perspective}

\author{
  Cecilia Berardo \\
  Department of Mathematics and Statistics\\
  FI-00014 University of Helsinki, Finland\\
  \texttt{cecilia.berardo@helsinki.fi} \\
   ORCID: 0000-0002-1729-3765
   \And
 Stefan Geritz \\
 Department of Mathematics and Statistics\\
 FI-00014 University of Helsinki, Finland\\
 \texttt{stefan.geritz@helsinki.fi} \\
 ORCID: 0000-0002-7865-3541    
 }

\begin{document}
\maketitle

\begin{abstract}
Theoretical models show that in a non-constant environment two predator species feeding on one and the same prey may coexist because the two species occupy different temporal niches: the one with the longer handling time has the advantage when prey is rare so that holding on to the same catch is the better option, while the species with the shorter handling time has the advantage when prey is common and easy to catch. In this paper we address the question whether a predator species with a handling time that is not fixed but decreases with prey density could be selective superior regardless of whether the prey is rare or common, as such predator would be able to occupy both temporal niches all by itself. 

To that end we study the Rosenzweig-MacArthur model with a modified Holling type II functional response with a density dependent handling time and a handling time dependent conversion factor. We find that the population dynamics tend to be richer than that of the standard model with fixed handling times because of the possibility of multiple positive equilibria and positive attractors. Increasing the strength of the density dependence eventually stabilises the population.

Using the framework of adaptive dynamics, we study the evolution of the strength of the density dependence. We find that a predator with even a weakly density dependent handling time can invade both monomorphic and dimorphic populations of predators with fixed handling times. Eventually the strength of the density dependence of the handling time evolves to a level where population cycles are lost, with that the possibility of predator coexistence as well.
\end{abstract}



\section{Introduction}\label{sec1}

It is a well established theoretical result that coexistence of multiple predator species competing for the same prey is possible, but only if the system exhibits non-equilibrium dynamics such as population cycles or chaos (see, e.g., \cite{koch1974competitive},  \cite{mcgehee1977some},  \cite{levins1979coexistence}, \cite{muratori1989remarks},  \cite{huisman1999biodiversity}, \cite{abrams2003dynamics},  \cite{liu2003relaxation}, \cite{wilson2004coexistence}). Additional conditions on the shape of the predator's functional and numerical response as functions of the prey density are typically required (\cite{armstrong1980competitive}). For example, \cite{abrams2002impact} showed that two predator species, both with a Holling type II functional response, can coexist provided that the handling times and the conversion factors (from number of prey captured into number of predator offspring produced) are sufficiently different for the two species.\\

Coexistence is possible because the two predator species occupy different ecological niches in the following sense: the species with the longer handling time has the advantage during the phase of the population cycle when the prey is rare so that holding on to the present catch is still the better option, even if it has lost much of its nutritious value. However, during the phase of the cycle when the prey is abundant, hanging on to the same corpse for a long time is not worthwhile anymore, because fresh prey is easy to come by, and so a shorter handling time is favourable.\\

The topic of coexistence of multiple predator species competing for the same prey was moved from the ecological theatre to an evolutionary context by \cite{geritz2007evolutionary}. Using the method of adaptive dynamics (\cite{metz1995adaptive}, \cite{geritz1997dynamics}, \cite{geritz1998evolutionarily}, \cite{geritz1999evolutionary}), they showed that an evolutionarily stable form of coexistence of two predators is not just ecologically feasible but can also evolve via evolutionary branching of the handling time in an initially monomorphic predator population. \\

In this paper we investigate the evolutionary consequences when the handling time is not constant, but depends on the prey density. In particular, if the handling time is a decreasing function of the prey population density, then such a predator might occupy both niches (i.e. short handling time and long handling time) at the same time, leaving no room for coexistence with others. The two main questions are:
\begin{enumerate}[label=(\roman*)]
\item\label{q1} Is branching (as seen in the model of \cite{geritz2007evolutionary}) prevented by the possibility of the evolution of a density dependent handling time?
\item\label{q2} Is evolutionarily stable coexistence of two predator species with fixed handling times (as seen in the model of \cite{geritz2007evolutionary}) prevented by the possibility of the evolution of a density dependent handling time? 
\end{enumerate}

These are basically questions about generalisation (density dependent handling time) versus specialisation (fixed handling time).
To address these questions, we study the evolution of the density dependence of handling time in the following generalisation of the Rosenzweig-MacArthur model (\cite{rosenzweig1963graphical}) for a single prey species $x$ and potentially multiple predator species with population size $y_i, i=1,...,n$ 
\begin{eqnarray}\label{RMA1}
\frac{dx}{dt}&=& r x \left(1- \frac{x}{K} \right) - \sum_{i=1}^n \frac{\beta x y_i}{1+\beta h_i(x)x}\\\label{RMA2}
\frac{dy_i}{dt}&=& \gamma(h_i(x))\frac{\beta x y_i}{1+h_i(x)\beta x}-\delta y_i, \qquad i=1,...,n.
\end{eqnarray}
When the predators are absent, the prey population grows according to the logistic equation. Each predator species $i$ has a Holling type II functional response with a prey density dependent handling time $h_i(x)$. The conversion factor $\gamma(h_i(x))$, modelling the number of newborn predators produced per captured prey, depends on the handling time and therefore indirectly on the prey density as well. \\

That the conversion factor depends on the handling time is an ecologically reasonable assumption: handling includes eating and digesting the captured prey. Therefore, the total amount of nutrient ingestion is likely to increase with the length of the handling time and, as a consequence, so does the conversion factor. How exactly the conversion factor depends on the handling time we explain in Section \ref{sec2} of this paper.\\

How the handling time would depend on prey density is maybe less clear. That, too, we explain in Section \ref{sec2}, where we derive the handling time from an underlying microscopic model for the interactions between individual predators and prey. In this way, we obtain a representation of the strength of the density dependence in terms of event rates on the microscopic level. These event rates, subject to mutation and selection, are the focal traits of which we study the evolution, using the theory of adaptive dynamics.\\

In Section \ref{sec3}, we analyse the population dynamics of the one-prey-one-predator resident population. The dynamics defines the selective environment in which a new mutation in the event rates underlying the density dependence of the handling time may or may not invade. \\

To find out which mutant strategies can invade a resident population with a given resident strategy, we study, in Section \ref{sec4}, the population dynamics of the one-prey-two-predator system, one predator being the resident and the other the mutant. In particular, we study the transversal stability of the boundary equilibrium where the resident is common but the mutant is absent. \\

The results in Section \ref{sec4} extend the analysis given by \cite{geritz2007evolutionary} on evolutionary branching in case of fixed handling time. The aim is to find out whether, in our more general model, evolution leads to the coexistence of two predator types each with a different fixed handling time or to a single predator type with a prey density dependent handling time. Finally, we investigate if a predator with prey density dependent handling time can invade the evolutionarily and convergence stable coexistence of two fixed handling times modelled by \cite{geritz2007evolutionary}.

\section{Derivation of the ecological model}\label{sec2}
Compared with the standard Rosenzweig-MacArthur predator-prey model (\cite{rosenzweig1963graphical}), the new elements in equations (\ref{RMA1}) and (\ref{RMA2}) are the prey density dependent handling time $h_i(x)$ and the handling time dependent conversion factor $\gamma(h_i(x))$. \\

The Holling type II functional response with a constant handling time can be derived from a system of fast state transitions for the predator between a searching state and a handling state on a time-scale during which total prey and predator densities remain constant. The fraction of time that an individual predator spends in the searching state determines the average \emph{per capita} rate of prey capture, i.e., the functional response. Using the mass-action assumption, the transition from searching to handling happens at a rate that is proportional to the prey density. If the rate of the reverse transition from handling to searching is assumed to be a constant, then the Holling type II functional response 

\begin{equation}
f(x,h) = \frac{\beta x}{1+\beta h x}
\end{equation}
is retrieved (\cite{metz2014dynamics}). \\

However, if the transition from handling to searching can also be brought about by a chance encounter with a live prey individual, then the transition rate will get a prey density dependent component. \cite{Berardo:2020aa} showed that in that case we derive the Holling type II functional response $f(x,h(x))$ with the density dependent handling time 

\begin{equation}
h(x) = \frac{1}{b x + c}
\end{equation}
where $bx+c$ is the transition rate from handling to searching with a density dependent part ($bx$) and a constant part ($c$), (see Appendix \ref{app1} for details on the mechanistic derivation and Figure \ref{fig:sec2} for the plot of the corresponding functional response). 
For $b=0$ we recover the ordinary Holling type II functional response, which saturates to a level $c$ as $x$ goes to infinity. For $b>0$, the response no longer saturates but keeps increasing (see Figure \ref{fig:sec2}). \\

In this paper we study the evolution of the pair $(b,c)$ as a two-dimensional trait, subject to mutation and selection using the methods of adaptive dynamics ( \cite{metz1992should}, \cite{metz1995adaptive}, \cite{geritz1997dynamics}, \cite{geritz1998evolutionarily},  \cite{geritz1999evolutionary}). \\

\begin{figure}[H]
\centering
\includegraphics[width=8cm,height=6cm]{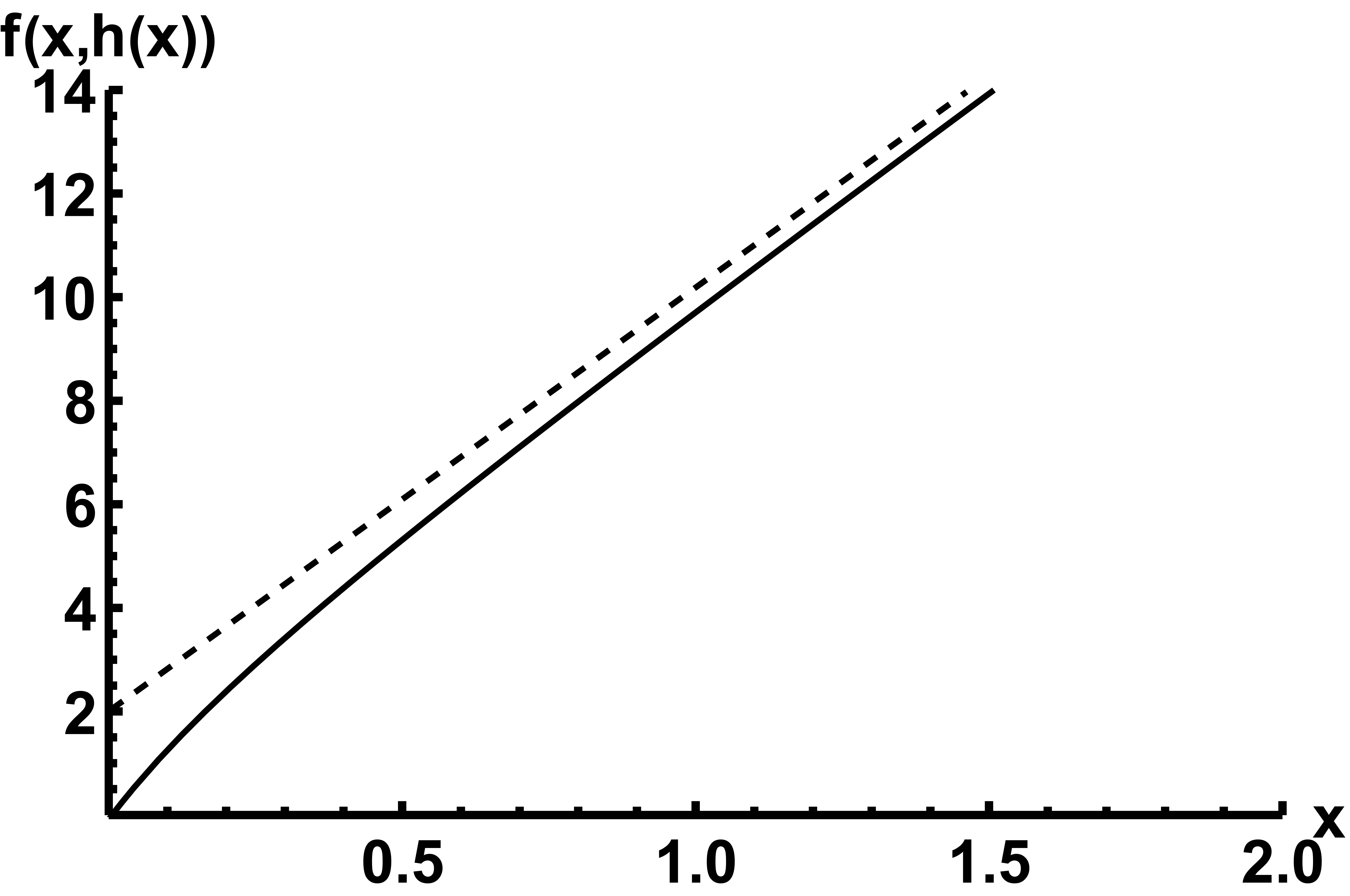}
\caption{The prey density dependent functional response derived by \cite{Berardo:2020aa}.}
\label{fig:sec2}
\end{figure}

The handling time dependent conversion factor $\gamma(h)$, expressed as the number of predator offspring that can be produced per captured prey, is calculated as

\begin{equation}\label{biggammadef}
\gamma(h)=\int_0^\infty \rho(\tau) e^{-\frac{1}{h}\tau} d\tau
\end{equation}
where $\rho(\tau)$ is the nutrient ingestion rate at $\tau$ time units after prey capture expressed in terms of number of predator offspring per unit of handling time. The exponential $e^{-\frac{1}{h}\tau}$ is the probability that prey handling is still going on at time $\tau$.\\

Note that there is a one-to-one relation between $\gamma$ and $\rho$ via the Laplace transform of $\rho$ and the function $\tilde\gamma$ defined such that $\tilde\gamma(\frac{1}{h})=\gamma(h)$. Specifically, $\tilde\gamma=\mathcal{L}\rho$ and $\rho=\mathcal{L}^{-1}\tilde\gamma$, where $\mathcal{L}$ is the Laplace transform operator and $\mathcal{L}^{-1}$ its inverse. Notwithstanding the one-to-one relation between $\gamma$ and $\rho$, from a modelling point of view the function $\rho$, being defined and interpreted on the individual level, is the more fundamental entity. We will assume that
\begin{enumerate}[label=(\roman*)]
\item\label{conddue} $\rho(\tau)$ is non-negative,
\item \label{conduno} $\rho(\tau)$ is integrable such that the integral over all positive $\tau$ is finite. 
\end{enumerate}

The first assumption can be interpreted as reproduction being nutrient rather than energy limited, so that the accumulated amount of resources acquired is a non-decreasing function of time in spite of possible energy costs of prey handling. The second assumption means that the total amount of nutrients that can be obtained from a single prey is finite. \\

As a first example, consider

\begin{equation}\label{rho1}
   \rho(\tau) = \gamma_0 \tau^k e^{-\frac{\tau}{\tau_0}} 
\end{equation}

for $k \ge 0$ and $\gamma_0,\tau_0>0$. Here the nutrient intake rate declines exponentially at large values of $\tau$. Such would be the case, for instance, if the search for nutritious parts of the captured prey is blind and random. The increase of $\rho(\tau)$ at small values of $\tau$ and $k>0$ (see Figure \ref{fig:examples}a) would occur if nutrient acquisition requires preparation such as dragging the prey to a safe location and opening up the carcass. The corresponding conversion factor $\gamma(h)$ as given by (\ref{biggammadef}) then becomes 

\begin{equation}\label{gamma1}
   \gamma(h) = \gamma_0 \Gamma(1+k) \left( \frac{h \tau_0}{h + \tau_0} \right)^{1+k} 
\end{equation}
(see Figure \ref{fig:examples}b), where $\Gamma(1+k) = k!$ is the gamma function for integer as well as non-integer values of $k$.\\

As a second example, consider

\begin{equation}\label{rho2}
   \rho(\tau) = \left\{
   \begin{array}{ll}
      \gamma_0 & \textrm{if $\tau_0 \le \tau \le \tau_1$} \\\\
      0 & \textrm{otherwise}
   \end{array}
   \right. 
\end{equation}
for $0 \le \tau_0 < \tau_1$ and $\gamma_0>0$. Here the nutrient intake rate is piecewise constant and becomes zero after a finite time when the carcass is fully spent (see Figure \ref{fig:examples}c). The initial zero intake rate for low values of $\tau$ and $\tau_0>0$ would occur for similar reasons as in the first example above. The corresponding $\gamma(h)$ becomes

\begin{equation}\label{gamma2}
   \gamma(h) = \gamma_0  \left( e^{-\frac{\tau_0}{h}} - e^{-\frac{\tau_1}{h}} \right) h.
\end{equation}
(see Figure \ref{fig:examples}d).\\

\begin{figure}[H]  
   \centering
   \includegraphics[width=4.5in]{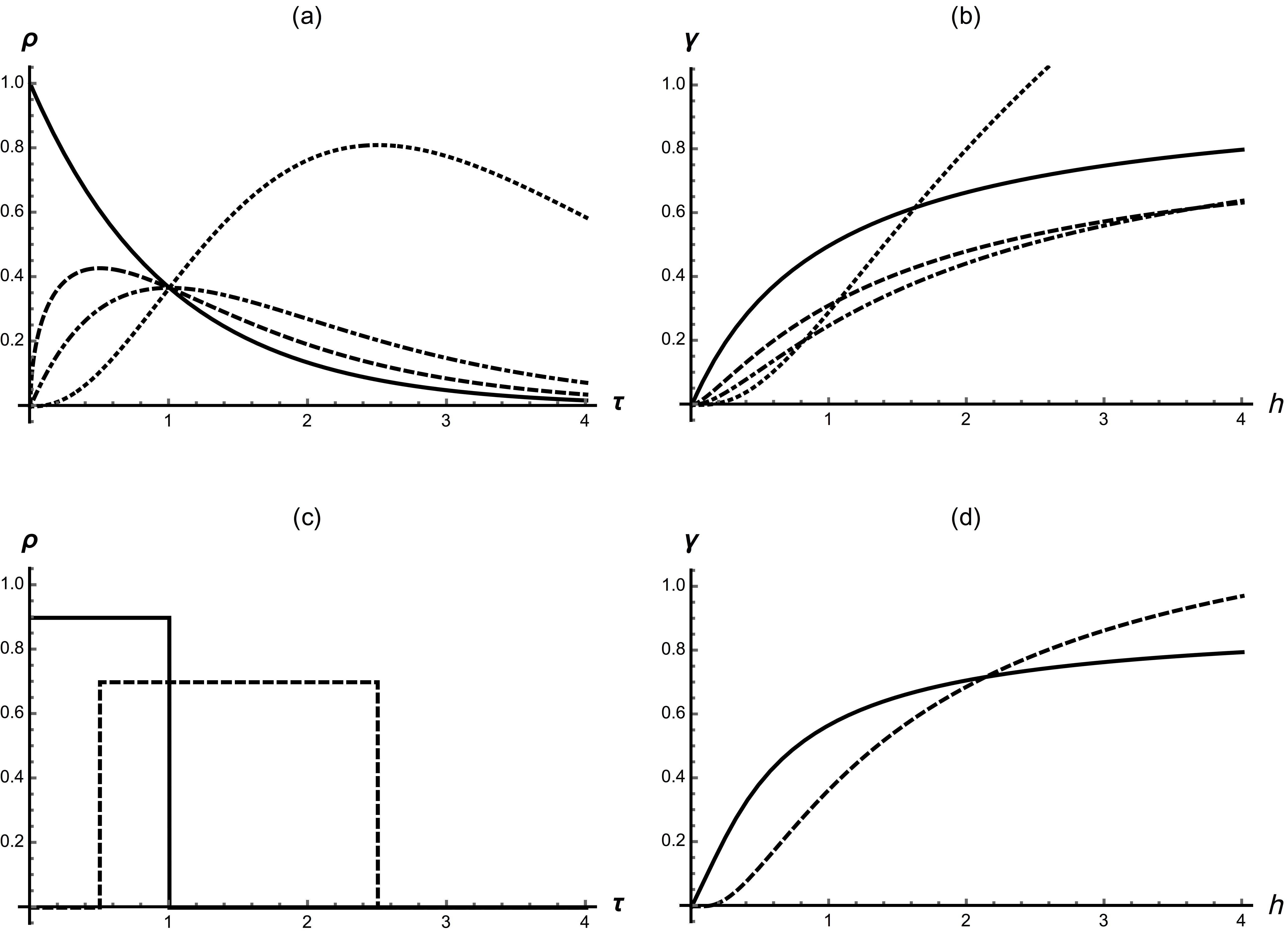} 
   \caption{Graphs $\rho(\tau)$ and corresponding $\gamma(h)$: Panels {\bf(a)} and {\bf(b)} correspond to respectively equations (\ref{rho1}) and (\ref{gamma1}) for $\gamma_0=1$, $\tau_0=1$, and $k=0$ (solid lines), $k=0.5$ (dashed), $k=1$ (dot-dashed), and $k=2.5$ (dotted). Panels {\bf(c)} and {\bf(d)} correspond to respectively equations (\ref{rho2}) and (\ref{gamma2}) for $\gamma_0=0.9$, $\tau_0=0$, $\tau_1=1$ (solid lines) and $\gamma_0=0.7$, $\tau_0=0.5$, $\tau_1=2.5$ (dashed)} 
   \label{fig:examples}
\end{figure}

Because of assumptions \ref{conddue} and \ref{conduno}, $\gamma$ is necessarily non-negative, non-decreasing and bounded as a function of  $h$. However, we emphasise that not every non-negative, non-decreasing and bounded function can be written as (\ref{biggammadef}) with $\rho$ satisfying the assumptions. As a counter-example we mention the $\gamma$ used by \cite{geritz2007evolutionary}, which was given numerically as an interpolating function through data points, and which, for that reason, is practically impossible to accurately inverse Laplace transform to obtain the corresponding $\rho$. However, their $\gamma$ is well approximated by the explicit function

\begin{equation} \label{gamma2007}
   \gamma(h) = \frac{A h^2}{C+h^2} + \frac{B h}{D+h}
\end{equation}
with $A= -3.06$, $B= 7.11$, $C=1.09$ and $D=1.36$, at least within the relevant range of $h$ (see Figure \ref{fig:approx}a). This function is non-negative, non-decreasing and bounded. The corresponding $\rho$ is 

\begin{equation} \label{rho2007}
   \rho(\tau) = \frac{A}{\sqrt C}\sin\left( \frac{\tau}{\sqrt C} \right) + B e^{-\tau}
\end{equation}
which, however, is periodically negative (see Figure \ref{fig:approx}b) and moreover cannot be integrated over all positive $\tau$, and therefore violates assumptions \ref{conddue} and \ref{conduno}. This illustrates that it is important to first model the nutrient acquisition in terms of the intake rate $\rho$ and from that derive the conversion factor $\gamma$. If one follows the opposite order and first chooses the $\gamma$, then there is no guarantee that it has an interpretation in terms of nutrient acquisition on the level of the individual prey and predator. \\

The above also presents a problem for the present article: how can we investigate under what conditions a hypothetical predator species with a density dependent handling time can invade the evolutionarily stable coexistence in the model of \cite{geritz2007evolutionary} if their $\gamma$ cannot exist in our present model? To solve this conundrum we will instead use 

\begin{equation} \label{rho2021}
   \rho(\tau) = \gamma_0 e^{-\frac{\tau}{\tau_0}} +\gamma_1 \tau e^{-\frac{\tau}{\tau_1}}
\end{equation}
with $\gamma_0=13.5651$, $\gamma_1=0.100023$, $\tau_0=0.117128$, $\tau_1=20$ (see Figure \ref{fig:approx}d). Both terms of $\rho$ have the form of (\ref{rho1}) with $k=0$ and $k=1$, respectively representing two kinds of resources: one that is readily available (first term) and one that requires some initial effort before it becomes accessible (second term). The corresponding $\gamma$ becomes

\begin{equation} \label{gamma2021}
   \gamma(h) = \gamma_0 \frac{h \tau_0}{h+\tau_0} + \gamma_1 \left( \frac{h \tau_1}{h+\tau_1} \right)^2
\end{equation}
(see Figure \ref{fig:approx}c). This $\gamma$ produces the same evolutionary behaviour as in \cite{geritz2007evolutionary} locally in the neighbourhood of the branching point $h=1.6$, but has a clear interpretation on the individual level in terms of nutrient acquisition. We come back to this in Section \ref{sec4}.\\

In order to obtain a better fit of the function by \cite{geritz2007evolutionary} and partially solve the numerical errors at large values of $h$ in the simulations (see Section \ref{sec4}), we control $\gamma$ for low and large values of $h$. Given the limited number of free parameters, it is clear that the function in (\ref{gamma2021}) is not suitable for this procedure. Therefore, we consider the six free-parameters function 

\begin{equation}\label{rho2021-2}
   \rho(\tau) = \gamma_{1} \tau^{k_1} e^{-\frac{\tau}{\tau_{1}}} +\gamma_{2} \tau^{k_2} e^{-\frac{\tau}{\tau_{2}}}.
   \end{equation}
Here the nutrient intake is obtained again from two resources which require some effort to be extracted from the prey carcass. The corresponding conversion factor becomes

\begin{equation} \label{gamma2021-2}
\gamma(h) =\gamma_1 \left(\frac{1}{h}+\frac{1}{\tau_1} \right)^{-1-k_1} \Gamma(1+k_1)+\gamma_2  \left(\frac{1}{h}+\frac{1}{\tau_2} \right)^{-1-k_2} \Gamma(1+k_2).
\end{equation}   
We fit $\gamma$ to the data points by \cite{geritz2007evolutionary} so that the local evolutionary dynamics around the branching point $h=1.6$ is the same. Then we fix an upper limit for $\gamma$ at $h\rightarrow\infty$ and we control the function at a small value of $h$, $h=0.4$, to obtain the set of parameters: $\gamma_1=42.691197$, $\gamma_2=0.014531$,$\tau_1=0.087297$, $\tau_2=0.884403$, $k_1=0.292140$ and $k_2=5.844259$ (see Figure \ref{fig:approx}e-f).\\

Finally, we fit the function in (\ref{gamma2021-2}) to the data points by \cite{geritz2007evolutionary} which produce the evolutionarily and convergence stable (ESS-) coexistence of the two fixed strategies $(h_1,h_2)=(1.5,4.5)$. We obtain the parameter values $\gamma_1=0.192742$, $\gamma_2=16.073108$,$\tau_1=27.416421$, $\tau_2=0.065766$, $k_1=0.146816$ and $k_2=0.146816$ (Figure \ref{fig:approx}g-h). We will use the function $\gamma$ in (\ref{gamma2021-2}) with the above parameter set to check if the ESS-dimorphism (resulting from evolutionary branching) of the two fixed handling times (with $h=\frac{1}{bx+c}$ and $b=0$) is attracting or repelling when we let the parameter determining density dependence, $b$, evolve to positive values.

\newpage

  \begin{figure}[H]  
   \centering
   \includegraphics[width=4in]{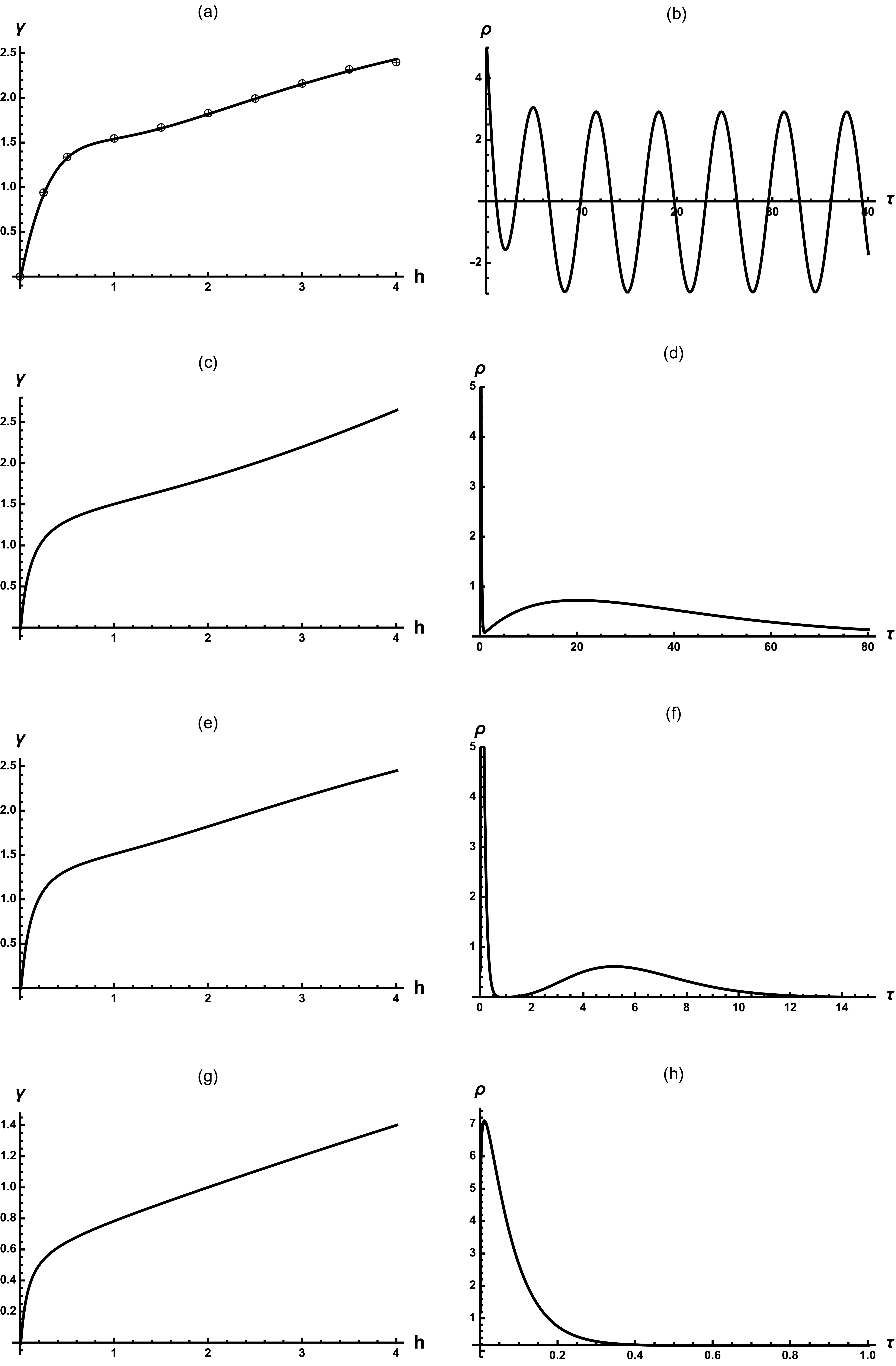} 
   \caption{{\bf(a)} The function $\gamma$ from equation (\ref{gamma2007}) approximating the function used by \cite{geritz2007evolutionary} some reference points of which are indicated by circles, and {\bf(b)} the corresponding $\rho$ from equation (\ref{rho2007}). {\bf(c)} The $\gamma$ from equation (\ref{gamma2021}) and {\bf(d)} the corresponding $\rho$ from equation (\ref{rho2021}). {\bf(e)} The $\gamma$ from equation (\ref{gamma2021-2}) and {\bf(f)} the corresponding $\rho$ from equation (\ref{rho2021-2}) with parameter values $\gamma_1=42.691197$, $\gamma_2=0.014531$,$\tau_1=0.087297$, $\tau_2=0.884403$, $k_1=0.292140$ and $k_2=5.844259$. {\bf(g)} The $\gamma$ from equation (\ref{gamma2021-2}) and {\bf(h)} the corresponding $\rho$ from equation (\ref{rho2021-2}) with parameter values $\gamma_1=0.192742$, $\gamma_2=16.073108$,$\tau_1=27.416421$, $\tau_2=0.065766$, $k_1=0.146816$ and $k_2=0.146816$.} 

   \label{fig:approx}
\end{figure}

\begin{table}[h]
   \caption{List of model variables, parameters and functions} 
   \label{tab:param}
   \small \centering
   \begin{tabular}{m{2cm} m{8cm}} 
   \textbf{Symbol} & \textbf{Description}  \\ \hline \hline
   $x$ & prey density\\
   $y_i$ & density of the predator species of type $i$\\
   $r$ & prey logistic growth rate \\ 
   $K$ & carrying capacity of the prey \\
   $\beta$ & predator's capture rate \\
   $b$ & prey-density specific rate of quitting handling \\
   $c$ & constant rate of quitting handling \\
   $\delta$ & \emph{per-capita} natural mortality rate of the predators\\
   $\rho(\tau)$ & rate of nutrient retention at time $\tau$ since prey capture\\
   \end{tabular}
\end{table}

\section{Resident dynamics}\label{sec3}
In this section we study the population dynamics of the one-prey-one-predator population. We call this the resident population, which defines the selective environment where an initially rare mutant with a different value of $(b,c)$ may or may not invade. 
When only one predator type with strategy $(b,c)$ is present, the system in (\ref{RMA1}) and (\ref{RMA2}) becomes
\begin{eqnarray}\label{monoone}
\frac{dx}{dt}&=&rx\left(1-\frac{x}{K} \right)-\frac{\beta xy}{1+\frac{1}{(bx+c)}\beta x}\\ \label{monotwo}
\frac{dy}{dt}&=&\int_0^\infty \rho(\tau) e^{-(b x+c)\tau} d\tau \frac{\beta xy}{1+\frac{1}{(bx+c)}\beta x}-\delta y.
\end{eqnarray}
Table \ref{tab:param} gives a summary of all functions and parameters. By proper scaling of the variables, parameters and functions (see Appendix \ref{app4}) we can effectively fix the values of $r$, $K$ and $\beta$, so that only the parameters $b$, $c$ and $\delta$ as well as the parameters that determine the function $\rho$ and the corresponding $\gamma$ are still free. Thus, without loss of generality, we set $r=1$, $K=1$ and $\beta=1$. \\

The prey density dependent handling time in the predator's functional response hardly matters for the population dynamics in comparison with the model of \cite{rosenzweig1963graphical}. However, the dependence of the conversion factor, and through that on the prey density as well, potentially makes a big difference. The equilibrium equations are
\begin{eqnarray} \label{eqeq1}
0&=&x(1-x)-f(x,h(x))y\\ \label{eqeq2}
0&=&\gamma(h(x)) f(x,h(x)) y-\delta y
\end{eqnarray}
from which we observe that the trivial equilibrium $E_0=(0,0)$ always exists, and so does the predator-free equilibrium $E_1=(1,0)$. In contrast to the model by \cite{rosenzweig1963graphical}, however, there can be multiple positive equilibria, i.e., with both the prey and the predator present at positive densities. This is because $\gamma(h(x))$ is a monotonically decreasing function of $x$ while $f(x,h(x))$ is monotonically increasing, so that nothing can be said in general about their product (see Appendix \ref{app2}). Multiple positive equilibria may, e.g., give rise to an Allee effect in the predator density (see Figure \ref{fig:sec31}), which is not possible in the standard Rosenzweig-MacArthur model. Moreover, a potential lack of monotony of $\gamma(h(x))f(x,h(x))$ affects not only the potential number of positive equilibria, but also their stability (see Appendix \ref{app5}).\\

\begin{figure}[H]
\centering
  \includegraphics[width=0.45\textwidth]{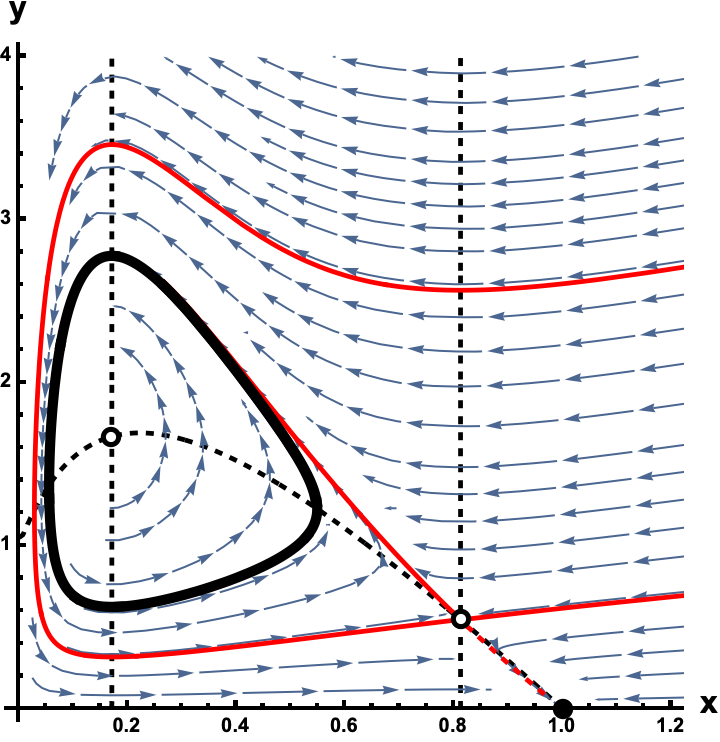}
\caption{Phase plane diagram for the system in (\ref{eqeq1}) and (\ref{eqeq2}) with $b=0.4$, $c=0.1$ (parameters of the function $h$), $\delta=1$, and $\rho(\tau)$ and corresponding $\gamma(h)$ as in (\ref{rho1}) and (\ref{gamma1}) with $k=5$, $\gamma_0=0.25$ and $\tau_0=1$. The dashed black lines are the prey and the predators isoclines. Open circles indicate unstable equilibria, and the black dot a stable equilibrium. The red orbits are the stable and unstable manifolds of the positive saddle and demarcate the basin of attraction of the limit cycle (thick black). The blue arrows indicate the flow of the dynamics. Orbits starting outside the basin of attraction of the limit cycle converge to the predator-free equilibrium. In particular, orbits starting with a low predator density lead to extinction of the predator.}
\label{fig:sec31}        
\end{figure}

We focus on a particular choice of the function $\gamma$ so that we can actually check whether $\gamma(h(x)) f(x,h(x))$ is monotonous or not. Specifically, we first choose $\rho$ and $\gamma$ from equations (\ref{rho2021}) and (\ref{gamma2021}) (see also Figure \ref{fig:approx}c-d), which for convenience are reproduced here
\begin{eqnarray} \label{another1}
   \rho(\tau)& =& \gamma_0 e^{-\frac{\tau}{\tau_0}} +\gamma_1 \tau e^{-\frac{\tau}{\tau_1}},\\ \label{another2}
   \gamma(h)& =& \gamma_0 \frac{h \tau_0}{h+\tau_0} + \gamma_1 \left( \frac{h \tau_1}{h+\tau_1} \right)^2,
\end{eqnarray}
with 

\begin{equation}\label{setparam1}
\gamma_0=13.5651,\quad \gamma_1=0.100023,\quad \tau_0=0.117128, \quad \tau_1=20.
\end{equation}
This is the function that gives evolutionary branching of handling time in the model of \cite{geritz2007evolutionary}, the robustness of which we will test if also density dependent handling times are allowed. \\

For this choice, we numerically confirmed that $\gamma(h(x))f(x,h(x))$ is monotonically increasing on the relevant interval $x\in[0,1]$ for all $(b,c) \in \mathbb{R}^2$. In particular, there is at most one positive equilibrium. When $b=0$, the population dynamics are essentially the same as in the standard Rosenzweig-MacArthur model and the interior equilibrium is stable if it lies on the decreasing part of the prey isocline and unstable if it lies on the increasing part. In the latter case there is a stable limit cycle.\\

When $b>0$ we find that for very low values of $c$ the interior equilibrium is stable (Figure \ref{fig:sec32}a and Figure \ref{fig:sec33}), but as soon as the system undergoes a fold bifurcation, bi-stability with a limit cycle occurs (Figure \ref{fig:sec32}b and Figure \ref{fig:sec33}). The bi-stability vanishes when the unstable limit cycle shrinks on the interior equilibrium (Figure \ref{fig:sec32}c and Figure \ref{fig:sec33}). After the subcritical Hopf bifurcation, the interior equilibrium is unstable and the system exhibits large oscillations as it converges to the stable limit cycle. Since both the prey and predator dynamics fluctuate between very low and large values, the change of variables given in Appendix \ref{app6} becomes fundamental to numerically approximate the large cycles. Further increasing $c$, the stable cycle shrinks to the unstable equilibrium and finally disappears via a supercritical Hopf bifurcation (Figure \ref{fig:sec32}d and Figure \ref{fig:sec33}).\\

As the parameters $b$ and $c$ are the evolutionary variables, Figure \ref{fig:sec34} shows which of the cases applies to where in the $(b,c)$-plane. Note that in this picture we are not able to capture the area of bi-stability of the interior equilibrium with the stable limit cycle and fold bifurcation line as this happens at vey low values of $c$. \\

We find the same population dynamics for the functions $\rho$ and $\gamma$ in (\ref{rho2021-2}) and (\ref{gamma2021-2}), more specifically for
\begin{eqnarray}\label{another12}
   \rho(\tau) &=& \gamma_{1} \tau^{k_1} e^{-\frac{\tau}{\tau_{1}}} +\gamma_{2} \tau^{k_2} e^{-\frac{\tau}{\tau_{2}}} \\\label{another22}
\gamma(h) &=&\gamma_1 \left(\frac{1}{h}+\frac{1}{\tau_1} \right)^{-1-k_1} \Gamma(1+k_1)+\gamma_2  \left(\frac{1}{h}+\frac{1}{\tau_2} \right)^{-1-k_2} \Gamma(1+k_2)
\end{eqnarray}   
with the two sets of parameter values given in Section \ref{sec2}

\begin{eqnarray}\nonumber
\gamma_1=42.691197,\quad \gamma_2=0.014531,\quad \tau_1=0.087297,\\\label{setparam2}
\tau_2=0.884403,\quad k_1=0.292140,\quad k_2=5.844259,
\end{eqnarray}
and

\begin{eqnarray}\nonumber
\gamma_1=0.192742,\quad \gamma_2=16.073108,\quad \tau_1=27.416421,\\\label{setparam3}
\tau_2=0.065766,\quad k_1=0.146816,\quad k_2=0.146816.
\end{eqnarray}

\newpage

\begin{figure}[H]
\centering
\begin{subfigure}[b]{0.35\textwidth}
\includegraphics[width=\textwidth]{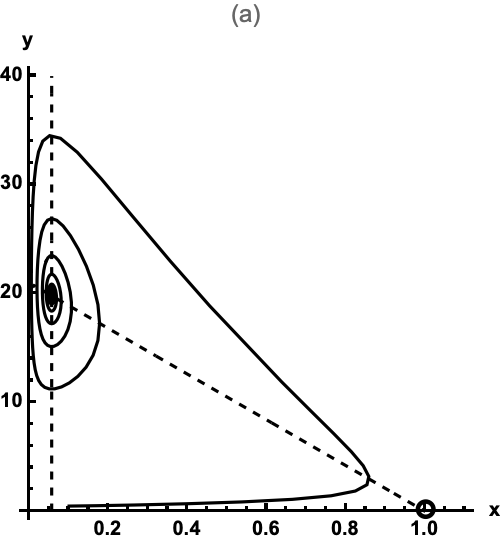}\hspace{0em}
\end{subfigure}
\begin{subfigure}[b]{0.35\textwidth}
\includegraphics[width=\textwidth]{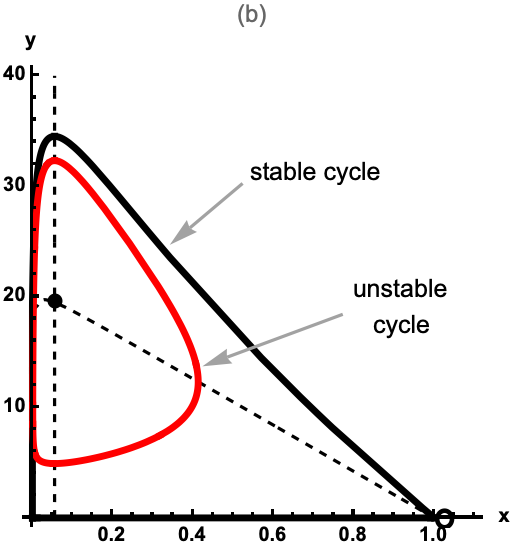}\hspace{0em}
\end{subfigure}
\begin{subfigure}[b]{0.35\textwidth}
\includegraphics[width=\textwidth]{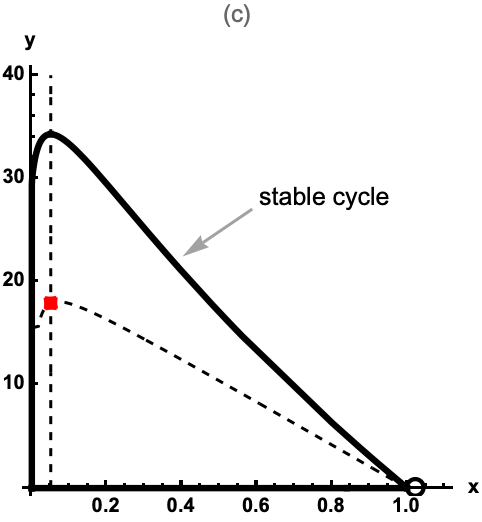}\hspace{0em}
\end{subfigure}
\begin{subfigure}[b]{0.35\textwidth}
\includegraphics[width=\textwidth]{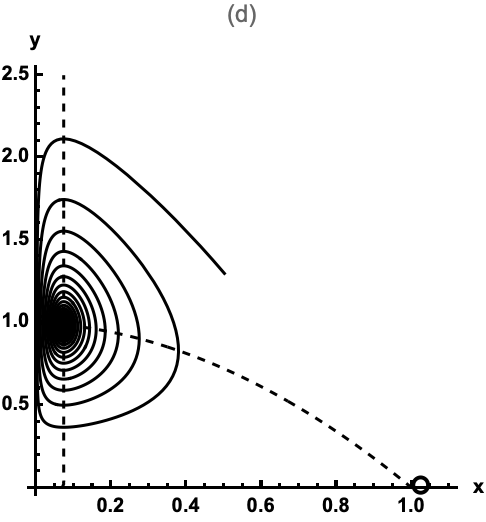}\hspace{0em}
\end{subfigure}
\caption{Typical dynamics of the model (\ref{monoone}) and (\ref{monotwo}) for low values of $c$ with $\rho$ and $\gamma$ as given in (\ref{another1}) and (\ref{another2}) and parameter values in (\ref{setparam1}). The dashed black lines are the prey and the predators isoclines. The black thin continuous lines in (a) and (d) transient orbits. The thick continuous lines in (b) and (c) are the limit cycles, black or red depending on stability. Open circles indicate unstable equilibria, and the black dot a stable equilibrium. (a): $c=0.0000001$; the interior equilibrium is stable; (b): $c=0.00005$; the system undergoes a fold bifurcation of a stable cycle and an unstable one; the interior equilibrium remains stable; (c): $c=0.0003$; the system undergoes a subcritical Hopf bifurcation; the unstable limit cycle shrinks on the interior equilibrium; (d): $c=1.1$; the system undergoes a supercritical Hopf bifurcation. } 
\label{fig:sec32}  
\end{figure}

\begin{figure}[H]
\centering
\begin{subfigure}[b]{0.6\textwidth}
\includegraphics[width=\textwidth]{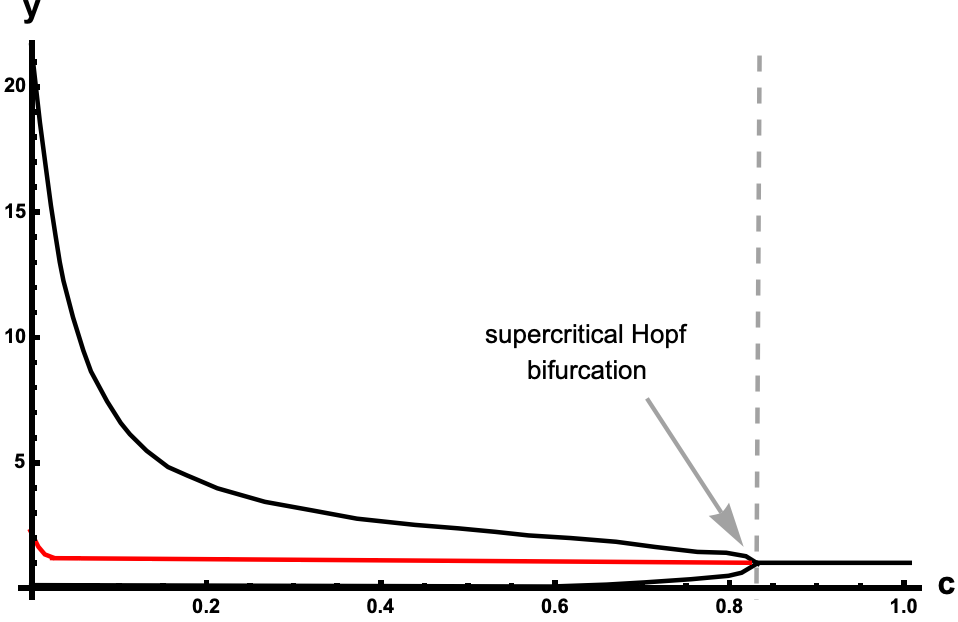}\hspace{0em}
\end{subfigure} 
\begin{subfigure}[b]{0.6\textwidth}
\includegraphics[width=\textwidth]{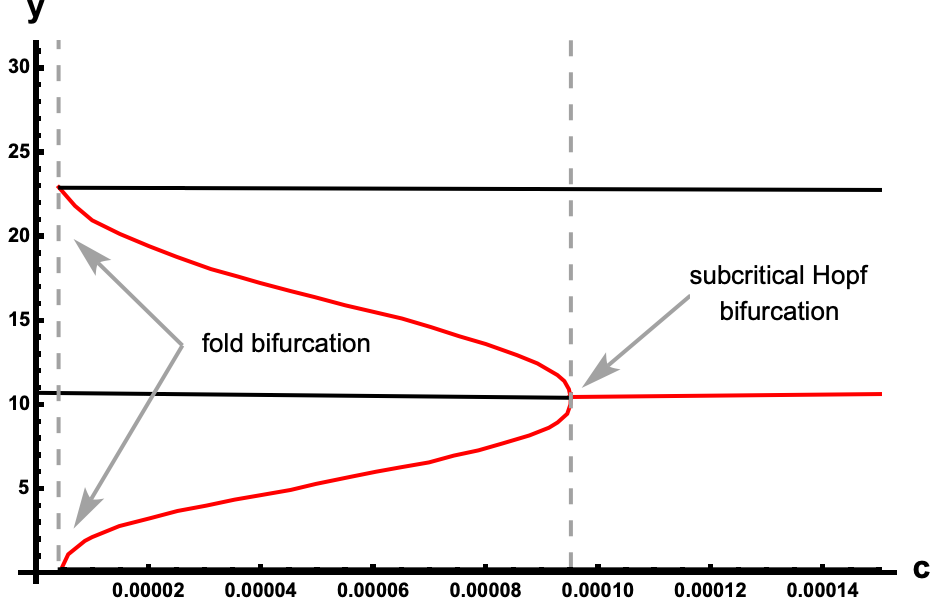}\hspace{0em}
\end{subfigure}
\caption{One-parameter bifurcation plot for the system in (\ref{monoone}) and (\ref{monotwo}) with $\rho$ and $\gamma$ as given in (\ref{another1}) and (\ref{another2}), parameter values in (\ref{setparam1}) and $b=0.05$. In red, the unstable equilibrium and unstable cycle. In black, the stable equilibrium and stable cycle. \emph{Top panel}: $c\in[0,1]$. \emph{Bottom panel}: $c\in[0,0.0002]$.} 
\label{fig:sec33}  
\end{figure}

\begin{figure}[H]
\centering
\begin{subfigure}[b]{0.3\textwidth}
\includegraphics[width=\textwidth]{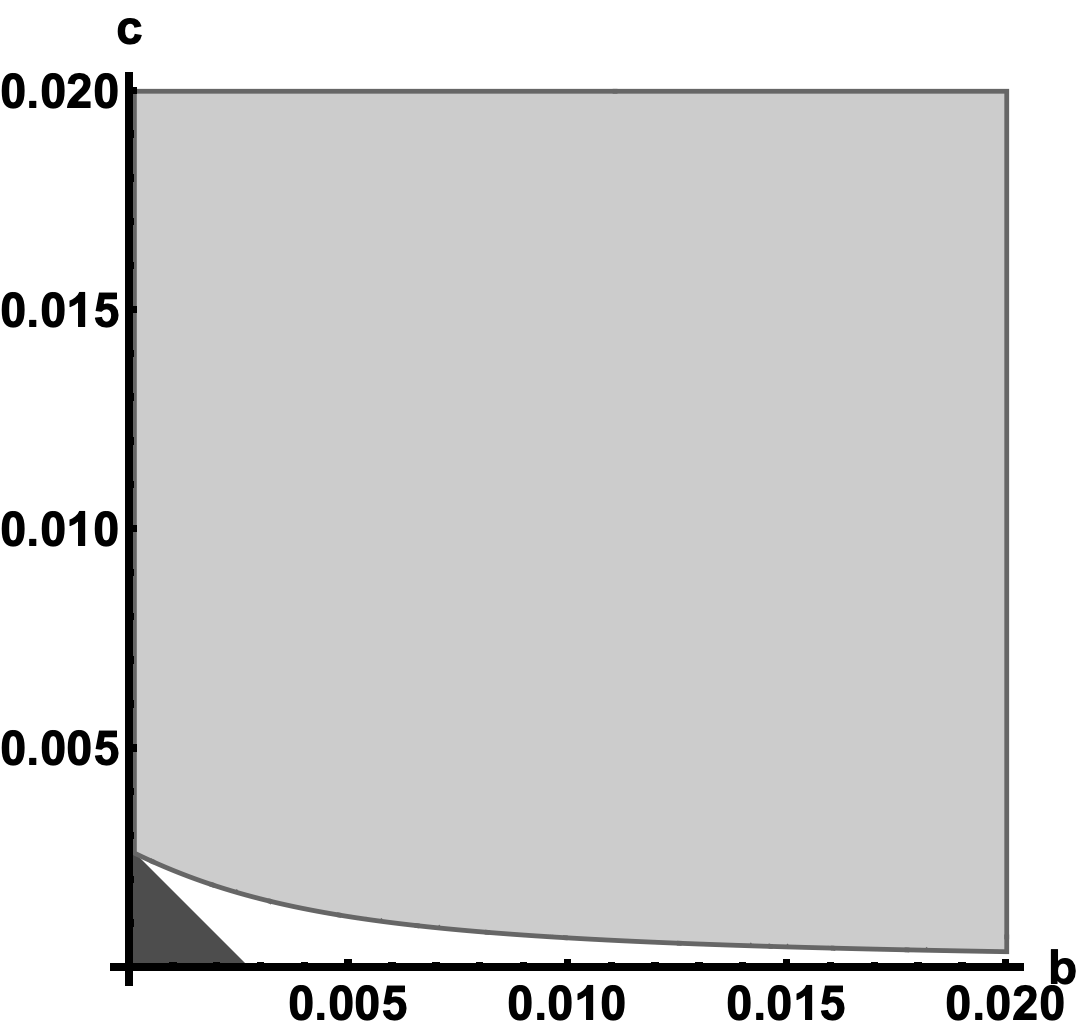}\hspace{0em}
\end{subfigure}
\begin{subfigure}[b]{0.3\textwidth}
\includegraphics[width=\textwidth]{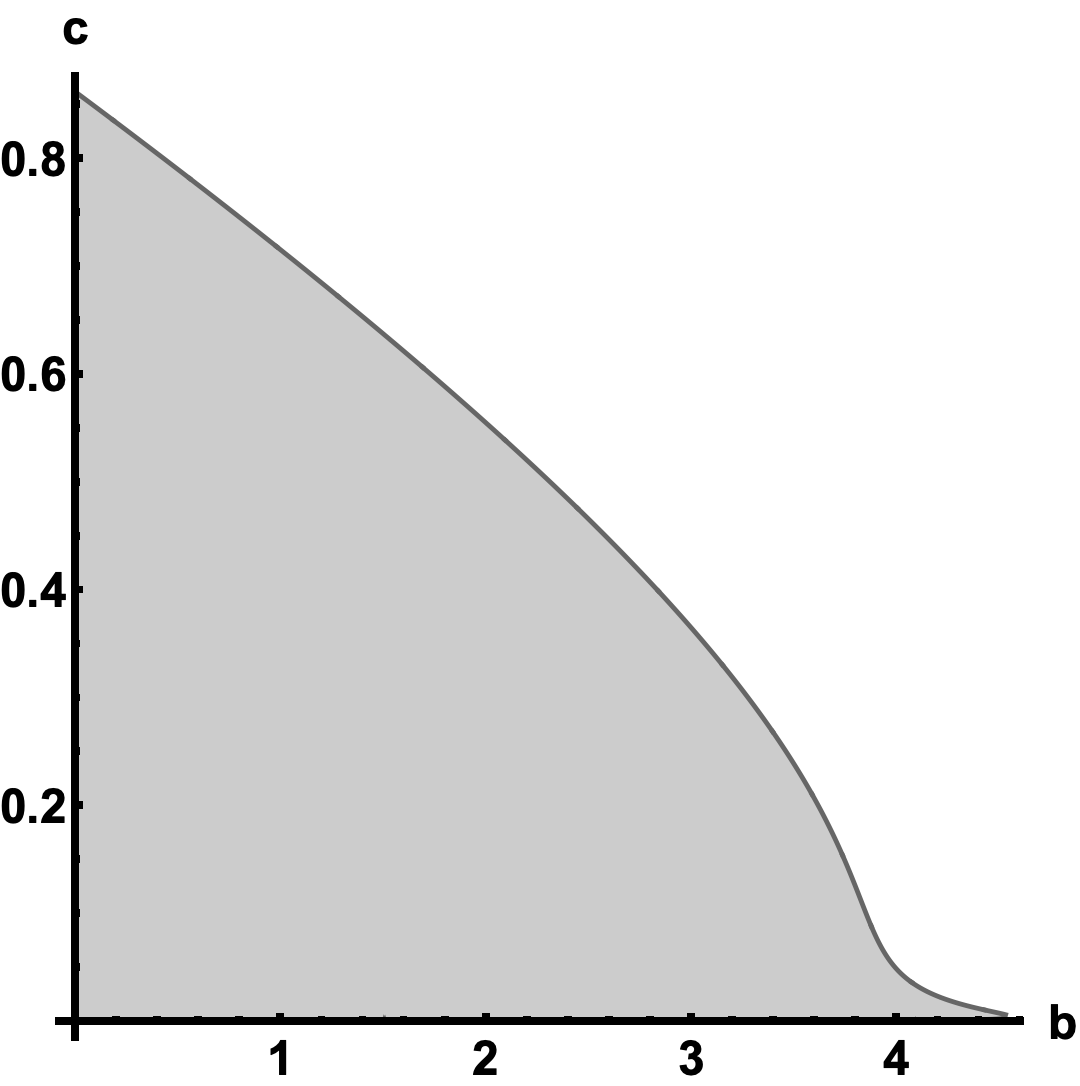}\hspace{0em}
\end{subfigure}
\begin{subfigure}[b]{0.3\textwidth}
\includegraphics[width=\textwidth]{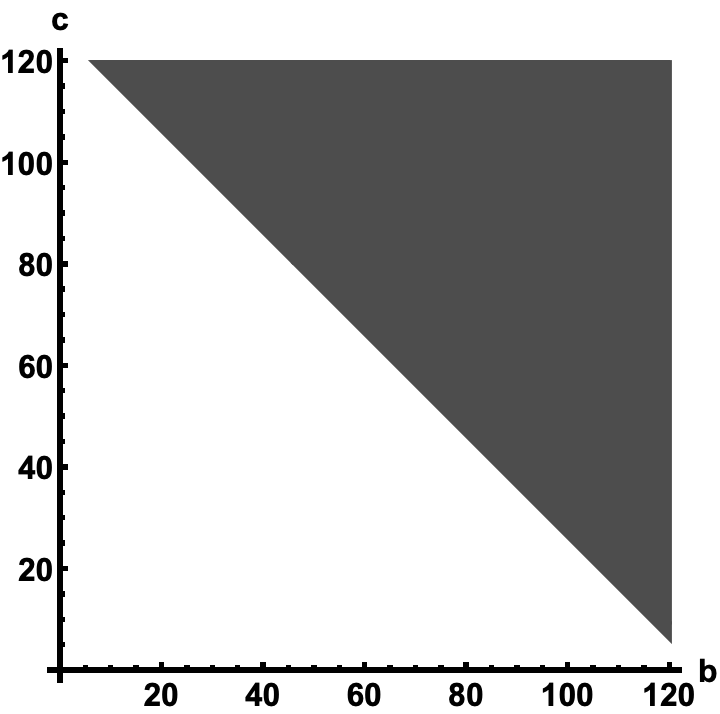}\hspace{0em}
\end{subfigure}
\begin{subfigure}[b]{0.3\textwidth}
\includegraphics[width=\textwidth]{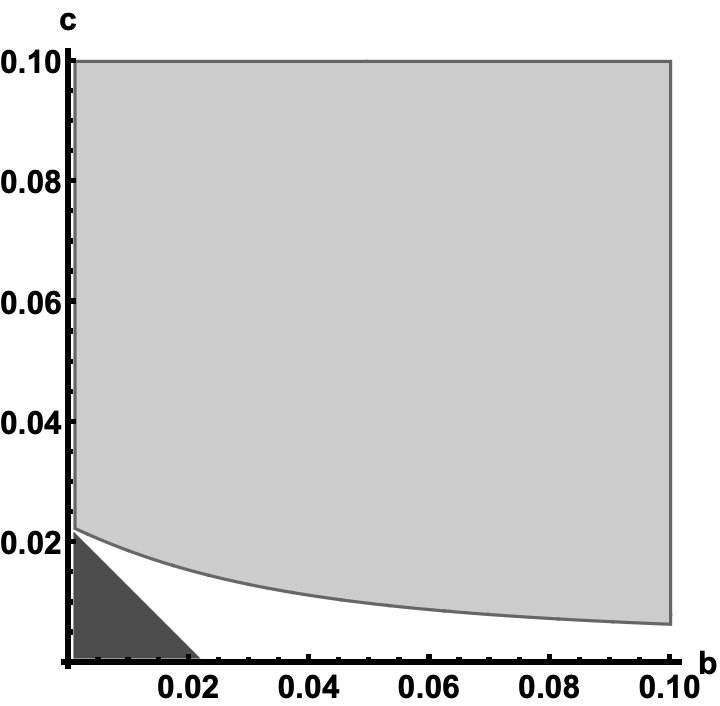}\hspace{0em}
\end{subfigure}
\begin{subfigure}[b]{0.3\textwidth}
\includegraphics[width=\textwidth]{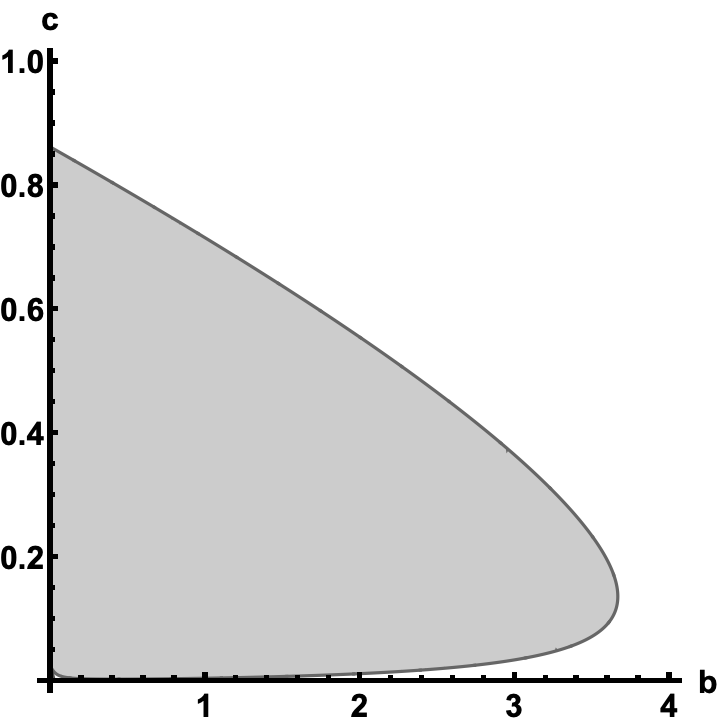}\hspace{0em}
\end{subfigure}
\begin{subfigure}[b]{0.3\textwidth}
\includegraphics[width=\textwidth]{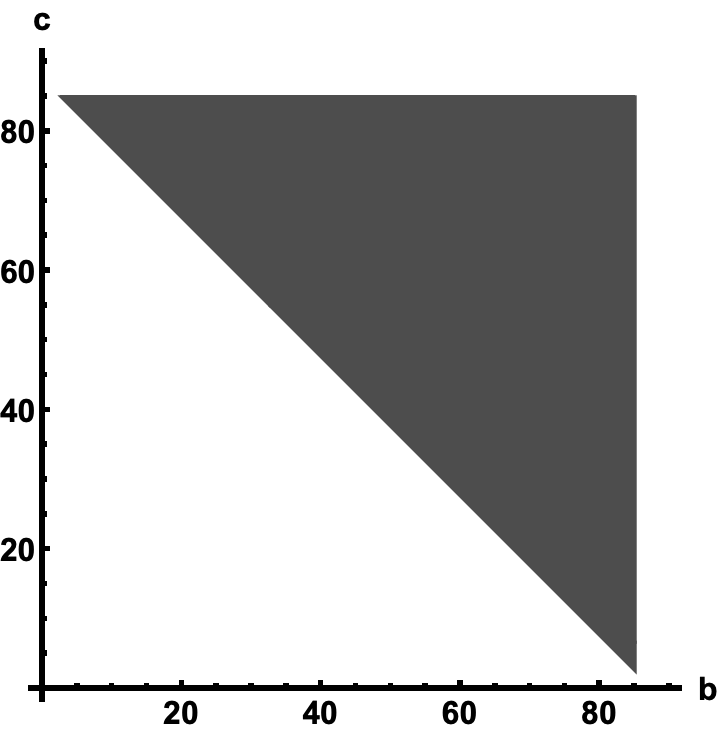}\hspace{0em}
\end{subfigure}
\begin{subfigure}[b]{0.3\textwidth}
\includegraphics[width=\textwidth]{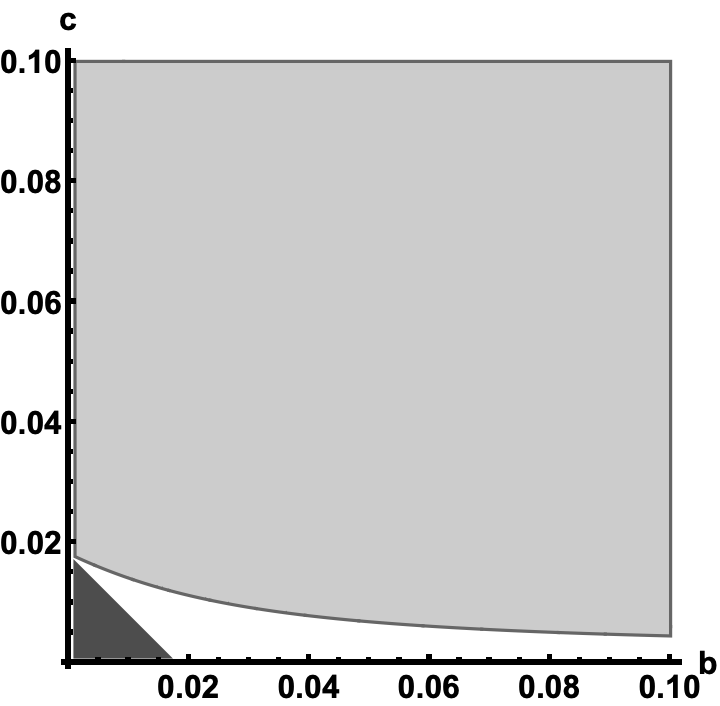}\hspace{0em}
\end{subfigure}
\begin{subfigure}[b]{0.3\textwidth}
\includegraphics[width=\textwidth]{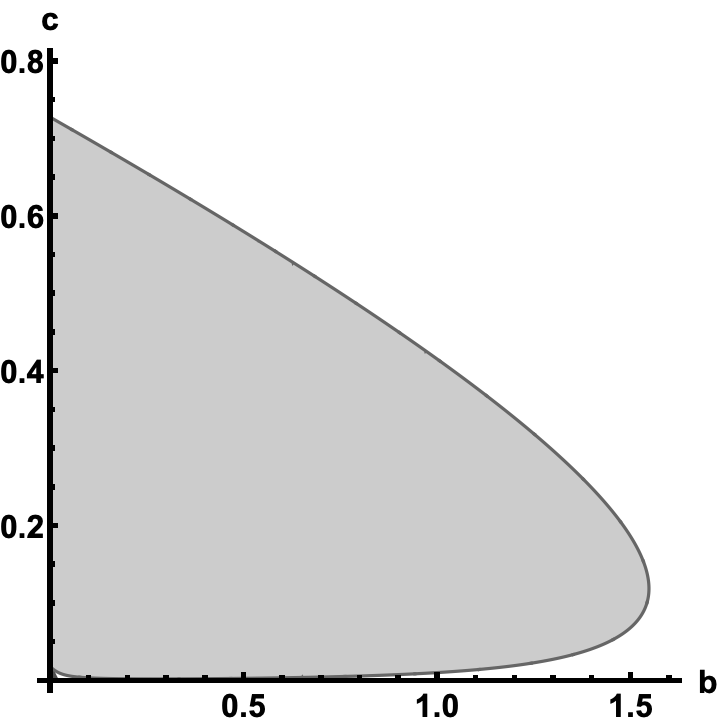}\hspace{0em}
\end{subfigure}
\begin{subfigure}[b]{0.3\textwidth}
\includegraphics[width=\textwidth]{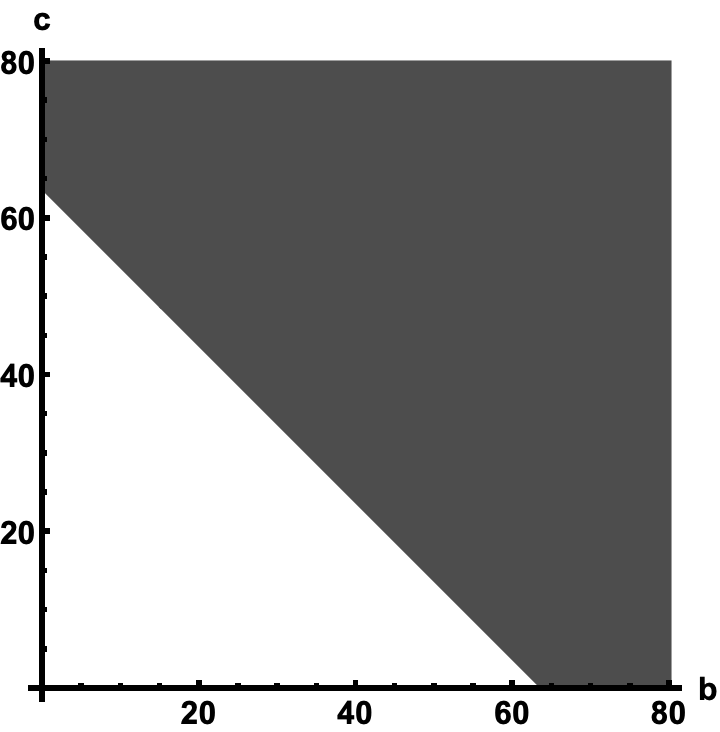}\hspace{0em}
\end{subfigure}
\caption{Two-parameters bifurcation plots for the system in (\ref{monoone}) and (\ref{monotwo}) and for different choices of $\rho$ and $\gamma$: {\em top row} for $\rho$ and $\gamma$ as given in (\ref{another1}) and (\ref{another2}) with (\ref{setparam1}); {\em middle row} for $\rho$ and $\gamma$ as given in (\ref{another12}) and (\ref{another22})  with (\ref{setparam2}); {\em bottom row} for $\rho$ and $\gamma$ as given in (\ref{another22}) and (\ref{another22}) with with (\ref{setparam3}). The columns show different scales: {\em left column} small scale; {\em middle column} intermediate scale; {\em right column} large scale. Different shadings refer to qualitatively different types of population dynamical behaviour: {\em dark grey} predator dies out; {\em light grey} stable limit cycle; {\em white} predator-prey stable equilibrium. Bi-stability of a stable cycle and a stable equilibrium as illustrated in Figure 6 occurs on a scale that is too small to be seen even on the smallest scales in left column. }
\label{fig:sec34}        
\end{figure}

\section{Evolutionary dynamics}\label{sec4}

We investigate the evolution of the density-dependent handling time $\frac{1}{bx+c}$ using the framework of adaptive dynamics. In particular, for the definitions of the \emph{invasion fitness} and \emph{selection gradient}, we refer to the work by \cite{metz1992should}, \cite{metz1995adaptive}, \cite{geritz1997dynamics}, \cite{geritz1998evolutionarily} and \cite{geritz1999evolutionary}. We assume that the evolution of the two-dimensional trait $(b,c)$ happens through a sequence of invasion events and replacements of the resident trait by the mutant one. The small mutation rate allows for the time scale separation between the ecological dynamics and the evolutionary dynamics of the traits, hence we assume that the population reaches an ecological attractor before the next mutant can invade and the resident population is monomorphic most of the time. \\

The ecological feedback environment $E$ is determined by the resident population at the equilibrium and is not affected by the rare mutant. The growth of the rare mutant depends on the ecological environment and if the mutant can invade the resident strategy, but not vice-versa, this will lead to the replacement of the resident trait with the mutant one. By the Tube Theorem (\cite{geritz2002invasion}), we know that if the rare mutant invades the resident dynamics, there will be no change in the resident attractor (the new resident population will \emph{inherit} the resident attractor). However, in case of a catastrophic discontinuous bifurcation (such as a saddle node bifurcation or an homoclinic bifurcation) and multiple attractors for the ecological dynamics, the successful invader can go extinct or switch resident attractor. This scenario cannot be excluded when the system in (\ref{eqeq1}) and (\ref{eqeq2}) has more than one positive and real solution and more than one coexistence equilibrium appear. \\

At the resident dynamics equilibrium, the ecological environment permits one-dimensional representation and is a one-dimensional vector with component the prey population $E=\{x\}$. We define the instantaneous \emph{per capita} growth rate for a predator type $i$ in the resident environment
as
\begin{eqnarray}\nonumber
\frac{1}{y_{i}}\frac{dy_{i}}{dt}&=&\frac{d(\log y_{i})}{dt}=\\\nonumber
&=& \int_0^\infty \rho(\tau) e^{-(b_{i} E+c_{i})\tau} d\tau \frac{\beta E }{1+\frac{1}{(b_{i}E+c_{i})}\beta E}-\delta=\\\label{insta}
&=& \gamma(h_{i}(E)) f(E,h_{i}(E))-\delta
\end{eqnarray}
The concentrations of different resources (i.e. prey densities) or other variables, such as the predator densities, are the main factors behind species diversity (see  \cite{macarthur1964competition}, \cite{tilman1982resource}, \cite{geritz1997dynamics}), and by the principle of competitive exclusion it follows that the dimensionality of the environment sets a theoretical upper limit to the number of phenotypes which can possibly coexist when the resident population attains its ecological equilibrium. Therefore, the coexistence of multiple resident types is excluded in case of convergence to the interior equilibrium. Under different conditions, such as convergence to the stable limit cycle, the coexistence of multiple consumers utilising the same resource could occur. Therefore, throughout the paper the focus is on the evolutionary dynamics in case of convergence to the resident periodic attractor. \\

When we introduce a mutant strategy $(b_{m},c_{m})$, the resident-invader dynamics is given by the following equations 
\begin{eqnarray}
\frac{dx}{dt}&=&r x\left( 1- \frac{x}{k}\right) - \sum_{i} \frac{\beta xy_i}{1+\frac{1}{(b_ix+c_i)}\beta x}\\
\frac{dy_i}{dt}&=&\int_0^\infty \rho(\tau) e^{-(b_i x+c_i)\tau} d\tau \frac{\beta xy_i}{1+\frac{1}{(b_ix+c_i)}\beta x}-\delta y_i\\
\frac{dy_{m}}{dt}&=&\int_0^\infty \rho(\tau) e^{-(b_{m} x+c_{m})\tau} d\tau \frac{\beta xy_{m}}{1+\frac{1}{(b_{m}x+c_{m})}\beta x}-\delta y_{m}.
\end{eqnarray}

A mutant strategy can invade the resident dynamics if it has positive \emph{invasion fitness} (note that, however, this will not always happen due to demographic stochasticity). The invasion fitness is defined by the exponential growth rate of the mutant population in the established resident population while the mutant is still rare. In the environment $E$ set by a single resident type with strategy $(b,c)$, the invasion fitness of the mutant population is given as the long-term average population growth rate

\begin{equation}\label{fit}
s_{(b,c)}(b_{m},c_{m})=\lim_{t\rightarrow\infty}\frac{1}{t}\int_0^t \left[ \gamma(h_{b_m,c_m}(E)) f(E,h_{b_m,c_m}(E))-\delta \right] dt.
\end{equation}
By definition, a fundamental property is that the invasion fitness of a resident type with strategy $(b,c)$ satisfies $s_{(b,c)}(b,c)=0$ at the demographic attractor. \\

The trait $(b,c)$ evolves in the direction of the local \emph{selection gradient}, i.e. the fitness derivative. Since the strategies are two-dimensional, the selection gradient is given by the two-dimensional vector with components the gradients with respect to $b$ and $c$

\begin{equation}\label{grad}
D(b,c)=\left(\frac{\partial s_{(b,c)}(b_{m},c_{m})}{\partial b_{m}}\Big|_{(b_{m},c_{m})=(b,c)}, \frac{\partial s_{(b,c)}(b_{m},c_{m})}{\partial c_{m}}\Big|_{(b_{m},c_{m})=(b,c)}\right).
\end{equation}
The selection gradient vanishes in the neighbourhood of an \emph{evolutionary singularity}. A singular strategy $(b,c)$ that no mutant trait can invade is \emph{evolutionary stable} (ESS, see the definition by \cite{smith1982evolution} in evolutionary game theory) and satisfies $s_{(b,c)}(b_{m},c_{m})<0$ for all pairs of trait values $(b_{m},c_{m})$ different from $(b,c)$. Furthermore, we define \emph{convergence stable} a singular strategy $(b,c)$ such that a mutant with strategy even closer to the singularity than the resident one can invade the resident dynamics. A singular strategy can be ESS and convergence stable, however if $(b,c)$ is convergence stable but not ESS, the system undergoes \emph{evolutionary branching} and two divergent strategies can coexist.\\

When the resident population is settled on the stable coexistence equilibrium $E_2=(x_2,y_2)$, the invasion fitness for the mutant strategy becomes
\begin{eqnarray}\nonumber
s_{(b,c)}(b_{m},c_{m})&=&\frac{x_2}{1+\frac{x_2}{b_mx_2+c_m}}\int_0^\infty \rho(\tau) e^{-(b_mx_2+c_m)\tau}d\tau-\delta=\\\label{fitteq}
&=&\gamma(h_{b_m,c_m}(x_2)) f(x_2,h_{b_m,c_m}(x_2))-\delta.
\end{eqnarray}
In this case, the selection gradients for the parameters $b$ and $c$ differ only by the factor $x_2$ and have same sign:
\begin{eqnarray}\nonumber
D(b,c)&=&\frac{\partial s_{E}(b_{m},c_{m})}{\partial c_{m}}\Big|_{(b_{m},c_{m})=(b,c)} \begin{pmatrix}
x_2\\
1
\end{pmatrix}=\\\nonumber
&=& \left[ \frac{x_2^2}{(x_2+bx_2+c)^2}\int_0^\infty \rho(\tau) e^{-(bx_2 +c)\tau} d\tau \right. \\\nonumber
&& \left. - \frac{x_2(b x_2+c)}{x_2+bx_2+c} \int_0^\infty \tau\rho(\tau) e^{-(bx_2+c)\tau}d\tau \right] \begin{pmatrix}
x_2\\
1
\end{pmatrix}.\\\label{gradeq}
\end{eqnarray}

Given the non-monotonicity of the fitness $s_{(b,c)}(b_{m},c_{m})$ in (\ref{fitteq}), the presence or absence of an optimisation (or pessimisation) principle must be verified (see \cite{metz2008does}). In Appendix \ref{app3}, we show that in case of constant conversion factor $\gamma(h)=\gamma_0$, a pessimisation principle is always verified and evolution minimises the prey density when the resident dynamics converges to the interior equilibrium. \\

On the other hand, when the conversion factor is a non-decreasing function of the handling time, the presence of a pessimisation principle is no longer trivial.
In Appendix \ref{app3} we suggest a graphical mean for checking non-monotonicity of $s_{(b,c)}(b_{m},c_{m})$ in case of non-constant $\gamma(h)$. We use the one-to-one relation between $x$ and $h_{b,c}$ such that if $s_{(b,c)}(b_{m},c_{m})$ is uniformly monotonic in $h_{b_m,c_m}$, so does with respect to $x$. Therefore, we formulate the problem in terms of $h_{b,c}$ and construct critical functions $\gamma_k(h_{b,c})$ with $k>0$ for the conversion factor $\gamma(h_{b,c})$. Sufficient and necessary conditions for non-uniformly monotonic fitness function are summarised in the following lemma:

\begin{lem}\label{lemmauno}
If there exist $k>0$ and $h_{b,c}\in\left(0, \frac{1}{c} \right)$ such that
\begin{enumerate}[label=(\roman*)]
\item\label{lemconduno} $\gamma (h_{b,c})=\gamma_k(h_{b,c}) $
\item \label{lemconddue} $\gamma '(h_{b,c})=\gamma_k'(h_{b,c})$
\item \label{lemcondtre} there exists $\varepsilon>0$ such that $\gamma''(h_0)$ does not change sign for every $h_0\in \left(h_{b,c}-\varepsilon, h_{b,c}+\varepsilon \right)$  
\end{enumerate}
Then the derivative $\frac{ds(h_{b,c})}{dh_{b,c}}$ changes sign at $h_{b,c}$.
\end{lem} 

Hence when Lemma \ref{lemmauno} is verified, an optimisation (or pessimisation) principle applies only locally and not uniformly. 
On the other hand, in a cycling population, the question is no longer an optimisation problem, but rather a frequency-dependent problem. When the system settles on a stable limit cycle, the invasion fitness of a mutant trait $(b_m,c_m)$ in the periodic environment generated by the resident strategy $(b,c)$ is given by the average instantaneous growth rate over the length of the cycle with period $t_{b,c}$:

\begin{equation}\label{fitcycle}
s_{(b,c)}(b_m,c_m)=\frac{1}{t_{b,c}}\int_0^{t_{b,c}}\left[\frac{x_{b,c}(t)}{1+\frac{x_{b,c}(t)}{b_m x_{b,c}(t)+c_m}}\int_0^\infty \rho(\tau) e^{-(b_m x_{b,c}(t)+c_m)\tau}d\tau-\delta \right]dt.
\end{equation}

In the same way, we define the selection gradient for the strategy $(b,c)$ in the periodic resident environment as

\begin{eqnarray}\nonumber
D(b,c)&=&\frac{1}{t_{b,c}}\int_0^{t_{b,c}} \left[ \frac{x_{b,c}(t)^2}{(x_{b,c}(t)+b x_{b,c}(t)+c)^2}\int_0^\infty \rho(\tau) e^{-(b x_{b,c}(t)+c)\tau} d\tau \right. \\ \nonumber
&&- \left. \frac{x_{b,c}(t)(b x_{b,c}(t)+c)}{x_{b,c}(t)+b x_{b,c}(t)+c} \int_0^\infty \tau \rho(\tau) e^{-(b x_{b,c}(t)+c)\tau} d\tau \right] \begin{pmatrix}
x_{b,c}(t)\\
1
\end{pmatrix} dt.\\\label{gradcycle}
\end{eqnarray}
Note that when the resident trait attains the stable limit cycle, the selection gradients for the parameter $b$ and $c$ no longer have same sign. Moreover, when the resident attractor is periodic, the dynamics can be studied only numerically.\\

We run numerical simulations with the software \emph{Mathematica}$^\circledR$ and apply the change of variables in Appendix \ref{app6} in case of one-prey-one-predator resident dynamics. We numerically integrate the equations with the command \emph{NDSolve} and use either the \emph{Event locator} method or the \emph{WhenEvent} controller and the Poincar\'e section to evaluate the convergence of the solution. \\

As a first step, we check the evolutionary dynamics for the one-dimensional predator trait $c=\frac{1}{h}$, corresponding to the fixed handling time by \cite{geritz2007evolutionary}, for $\rho$ and $\gamma$ as given in (\ref{another1}) and (\ref{another2}) and parameter values in (\ref{setparam1}), when $b=0$ and the resident population is at the periodic attractor.
In Figure \ref{fig:sec41}, we give the \emph{mutual invasibility plot} (MIP) for the predator trait $c$, that is we compute the sign of the invasion fitnesses $s_{(0,c_1)}(0,c_2)$ and $s_{(0,c_2)}(0,c_1)$ (corresponding to the signs in each area of the plot, respectively). 
We find the branching point $c=0.625$ (convergence stable but not evolutionary stable singularity, corresponding to $h=\frac{1}{c}=1.6$ in \cite{geritz2007evolutionary}), an evolutionary repeller at $c\approx 0.31$ (i.e. evolution leads away from this singular strategy) and an evolutionary attractor at $c\approx 0.06$ (ESS).\\

The area marked with $++$ in the MIP, where $s_{(0,c_1)}(0,c_2)>0$ and $s_{(0,c_2)}(0,c_1)>0$ and the strategies $c_1$ and $c_2$ are mutually invasible, corresponds to the set of possible coexisting strategies and we use evolutionary phase plane analysis to check the existence of dimorphic singularities and their stability. 
In particular, in the $++$ region above the main diagonal of the \emph{trait evolution plot} (TEP) in Figure \ref{fig:sec41}, we give the isoclines for the evolutionary dynamics in a dimorphic resident population with strategies $c_1$ and $c_2$, where the local fitness gradients 

\begin{eqnarray}
D_{1}(0,c_1,0,c_2)&=&\frac{\partial s_{(0,c_1,0,c2)}(b_m,c_m)}{\partial c_m}\Big|_{(b_m,c_m)=(0,c_1)},\\
D_{2}(0,c_1,0,c_2)&=&\frac{\partial s_{(0,c_1,0,c2)}(b_m,c_m)}{\partial c_m}\Big|_{(b_m,c_m)=(0,c_2)}
\end{eqnarray}
change sign and indicate by arrows the direction of evolution. In the area of coexistence above the main diagonal, the $c_1$-isocline connects to the $c_2$-extinction boundary with $s_{(0,c_1)}(0,c_2)=0$ at any point vertically above a singular strategy, while the $c_2$-isocline connects at any point of intersection with the vertical tangent (see the Appendix by \cite{geritz1999evolutionary} for more details). Note that for very small values of $c_1$ ($c_1\lessapprox 0.002$), we lose track of the evolutionary dynamics in the numerical simulations as it is too difficult to approximate the large cycles in the dimorphic resident population and then compute the mutant fitness gradient, thus we cannot find the $c_2$-isocline connection with the evolutionary boundary on the left hand side.
By symmetry, the same isoclines appear in the $++$ region below the main diagonal. \\

We observe that no evolutionary singular coalition occurs, but cycles of evolutionary branching and extinction are likely to happen. Let us start with a monomorphic population with strategy above the repeller, evolution first leads to $c=0.625$ where the system undergoes evolutionary branching. In the area of protected dimorphisms, the direction of evolution is up and to the left below the $c_2$-isocline and down and to the left above the $c_2$-isocline. Therefore, evolution follows the $c_2$-isocline up till the intersection with the extinction boundary, where the larger strategy of the dimorphism goes extinct and the remaining monomorphic population is below the repeller and thus moves to the ESS.\\

Different is the case with $\rho$ and $\gamma$ as given in (\ref{another12}) and (\ref{another22}). For both sets of parameter values in (\ref{setparam2}) and (\ref{setparam3}), the system evolves to a dimorphic coalition after branching (see Figure \ref{fig:sec42}, \emph{top panels}), $(c_1,c_2)=(0.271596,1.375601)$ and $(c_1,c_2)=(0.222222,0.666667)$ (corresponding to the evolutionary stable coexistence by \cite{geritz2007evolutionary}), respectively. From the TEPs it can be seen that the dimorphic singularity can be reached from most of the dimorphisms in the invasion cone, which is locally forward invariant, and therefore the dimorphic singularity is totally convergence stable. As the singularity represents also a combination of fitness maxima, the coalition is also evolutionary stable (see Figure \ref{fig:sec42}, \emph{bottom panels}).

\begin{figure}[H]
\centering
\begin{subfigure}[b]{.45\linewidth}
\includegraphics[width=6cm, height=6cm]{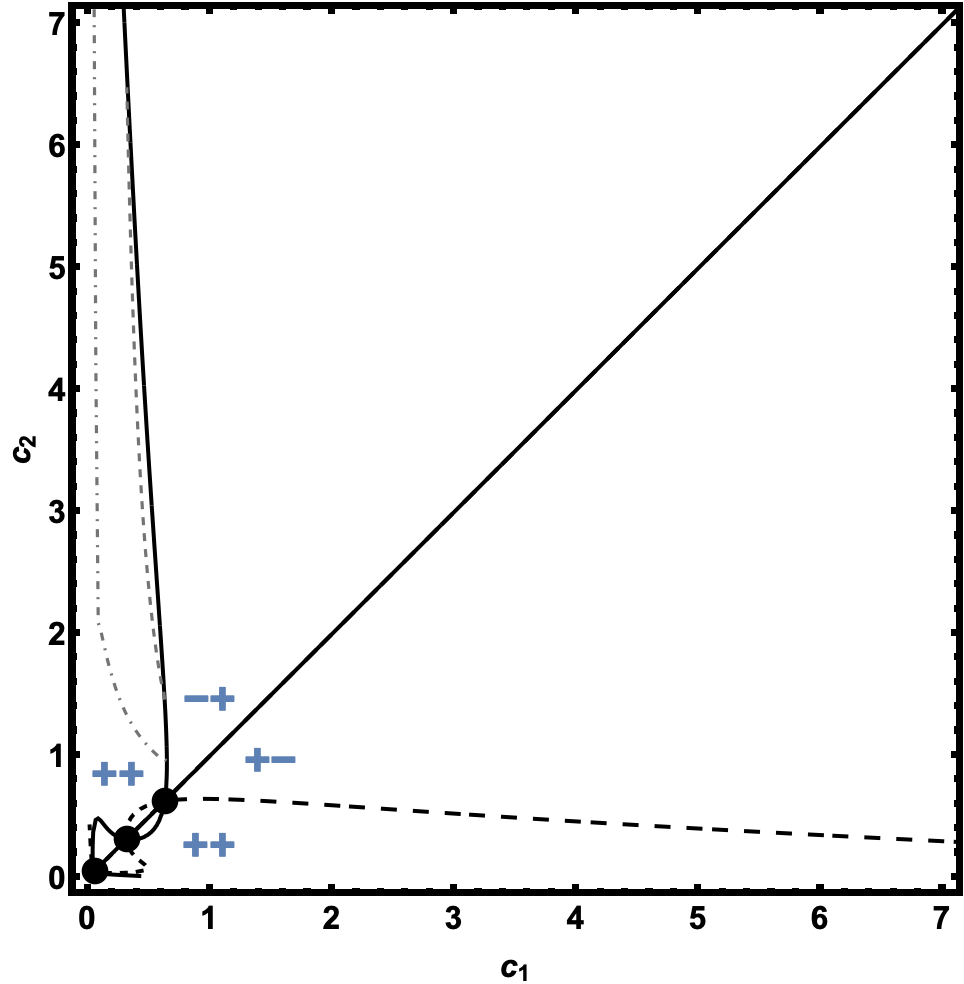}
\end{subfigure}
\begin{subfigure}[b]{.45\linewidth}
\includegraphics[width=6cm, height=6cm]{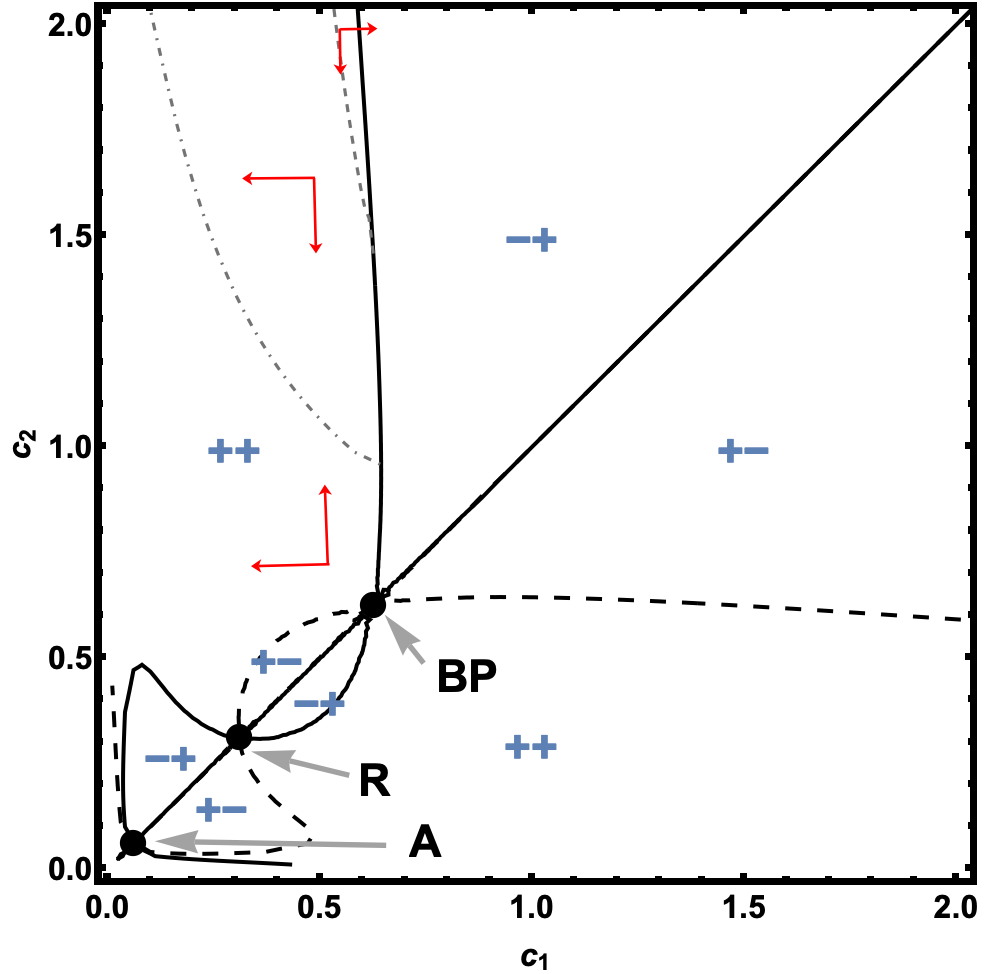}
\end{subfigure}
\caption{\emph{Left panel}: MIP for evolving $c$ and fixed $b=0$, for the system in (\ref{monoone}) and (\ref{monotwo}) and $\rho$ and $\gamma$ as given in (\ref{another1}) and (\ref{another2}) with parameter values in (\ref{setparam1}). \emph{Right panel}: corresponding TEP. Black continuous line: $c_2$-extinction boundary. Black dotted line: $c_1$-extinction boundary. BP: branching point. R: repelling singularity. A: attracting singularity. $++$: $s_{(0,c_1)}(0,c_2)>0$, $s_{(0,c_2)}(0,c_1)>0$, i.e. mutually invasible strategies.  $+-$: $s_{(0,c_1)}(0,c_2)>0$, $s_{(0,c_2)}(0,c_1)<0$, i.e. invasion and substitution. $-+$: $s_{(0,c_1)}(0,c_2)<0$, $s_{(0,c_2)}(0,c_1)>0$, i.e. the mutant strategy cannot invade. Grey dashed line: $c_1$-isocline. Grey dash-dotted line: $c_2$-isocline. Red arrows: evolutionary dynamics.}
\label{fig:sec41}        
\end{figure}

\begin{figure}[H]
\centering
\begin{subfigure}[b]{.45\linewidth}
\includegraphics[width=6cm, height=6cm]{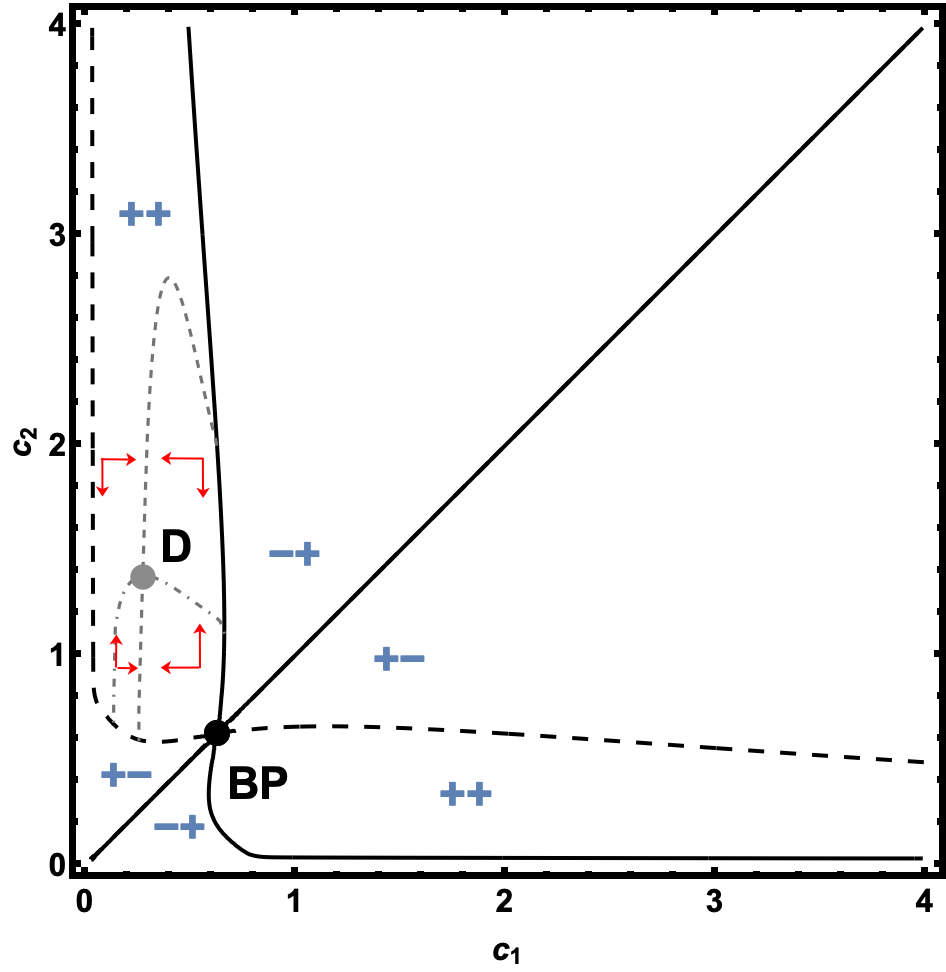}
\end{subfigure}
\begin{subfigure}[b]{.45\linewidth}
\includegraphics[width=6cm, height=6cm]{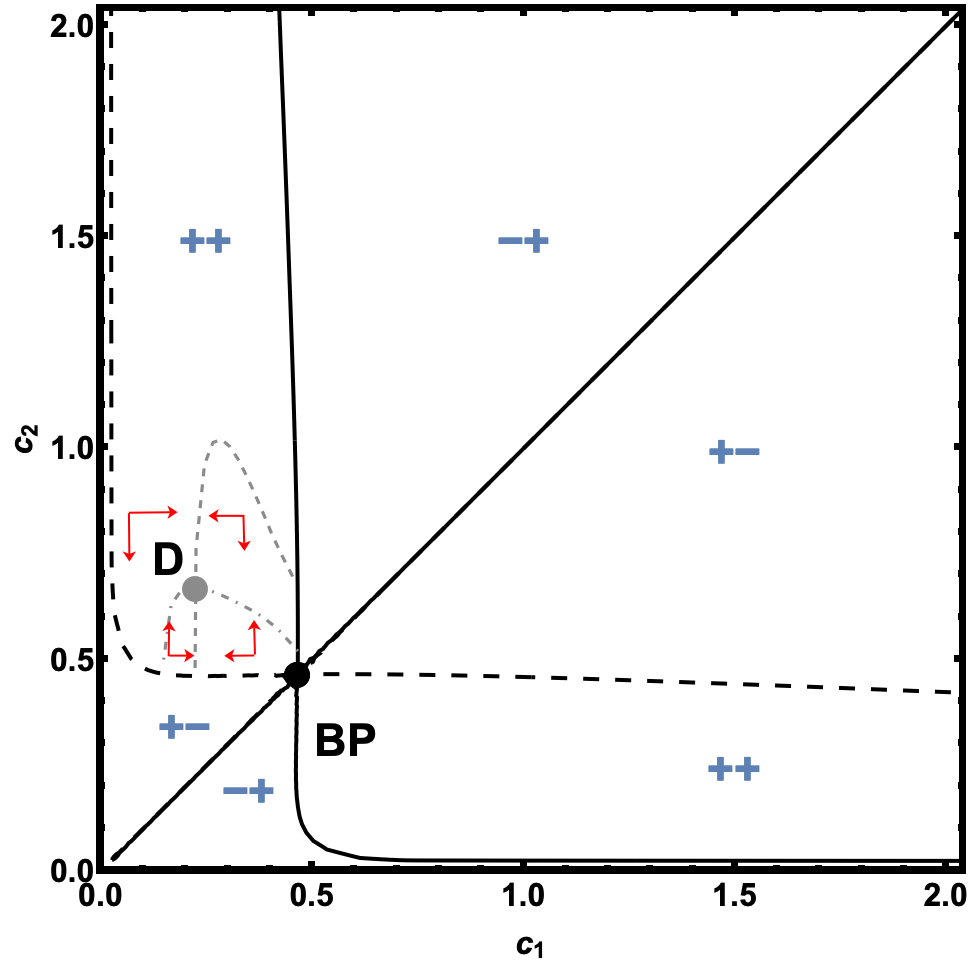}
\end{subfigure}
\begin{subfigure}[b]{.45\linewidth}
\includegraphics[width=5.5cm, height=3cm]{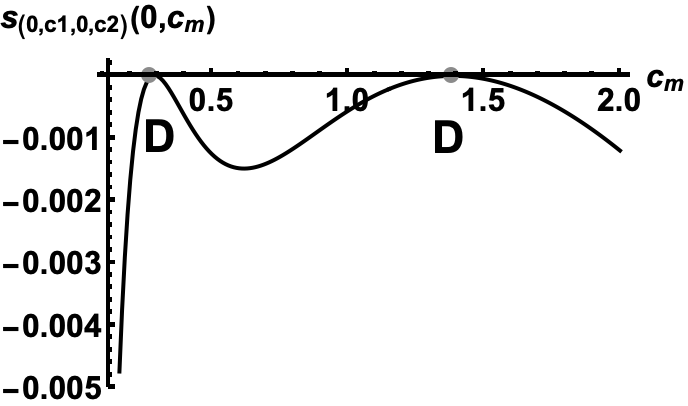}
\end{subfigure}
\begin{subfigure}[b]{.45\linewidth}
\includegraphics[width=5.5cm, height=3cm]{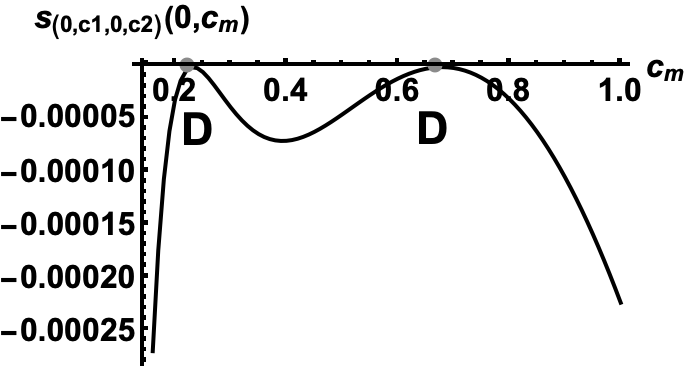}
\end{subfigure}
\caption{TEP for evolving $c$ and fixed $b=0$ and invasion fitness $s_{(0,c_1,0,c_2)}(0,c_m)$ of a mutant $c_m$ at the resident dimorphic coalition, for the system in (\ref{monoone}) and (\ref{monotwo}) and $\rho$ and $\gamma$ as given in (\ref{another12}) and (\ref{another22}). \emph{Left panels}: parameter values in in (\ref{setparam2}). \emph{Right panels}: parameter values in in (\ref{setparam3}). Black continuous line: $c_2$-extinction boundary. Black dotted line: $c_1$-extinction boundary. BP: branching point. D: dimorphic singularity. $++$: $s_{(0,c_1)}(0,c_2)>0$, $s_{(0,c_2)}(0,c_1)>0$, i.e. mutually invasible strategies.  $+-$: $s_{(0,c_1)}(0,c_2)>0$, $s_{(0,c_2)}(0,c_1)<0$, i.e. invasion and substitution. $-+$: $s_{(0,c_1)}(0,c_2)<0$, $s_{(0,c_2)}(0,c_1)>0$, i.e. the mutant strategy cannot invade. Grey dashed line: $c_1$-isocline. Grey dash-dotted line: $c_2$-isocline. Red arrows: evolutionary dynamics.}
\label{fig:sec42}        
\end{figure}

As a second step, we let $b$ evolve to positive values and give the evolutionary phase plane when the trait is two-dimensional. The aim is to investigate if the evolutionary branching point by \cite{geritz2007evolutionary} can be invaded by a mutant with density dependent strategy. In Figure \ref{fig:sec43}, one can see that in all cases evolution leads away from the branching point on the $c$-axis towards positive values of $b$ and destroys the possibility of evolutionary branching, as no intersection of the isoclines is present. While the dynamics for $\rho$ and $\gamma$ as given in (\ref{another1}) and (\ref{another2}) remains undefined when $c$ is very small, in case of $\rho$ and $\gamma$ in (\ref{another12}) and (\ref{another22}), the isoclines intersect on the HB-line and then coincide in the region of convergence to the stable equilibrium. In this particular scenario, the ecosystem evolves towards the most complex dynamics on the edge of stability. Similar behaviour was described by \cite{ellner1995chaos} in their conjecture on the evolution towards the edge of chaos and has found support in the works by \cite{dercole2006coevolution}, \cite{gragnani1998food}, \cite{rai2004chaos}, \cite{rai2006evolving}, \cite{rinaldi1999top}.

\begin{figure}[H]
\centering
\begin{subfigure}[b]{.5\linewidth}
\includegraphics[width=6.5cm, height=6.5cm]{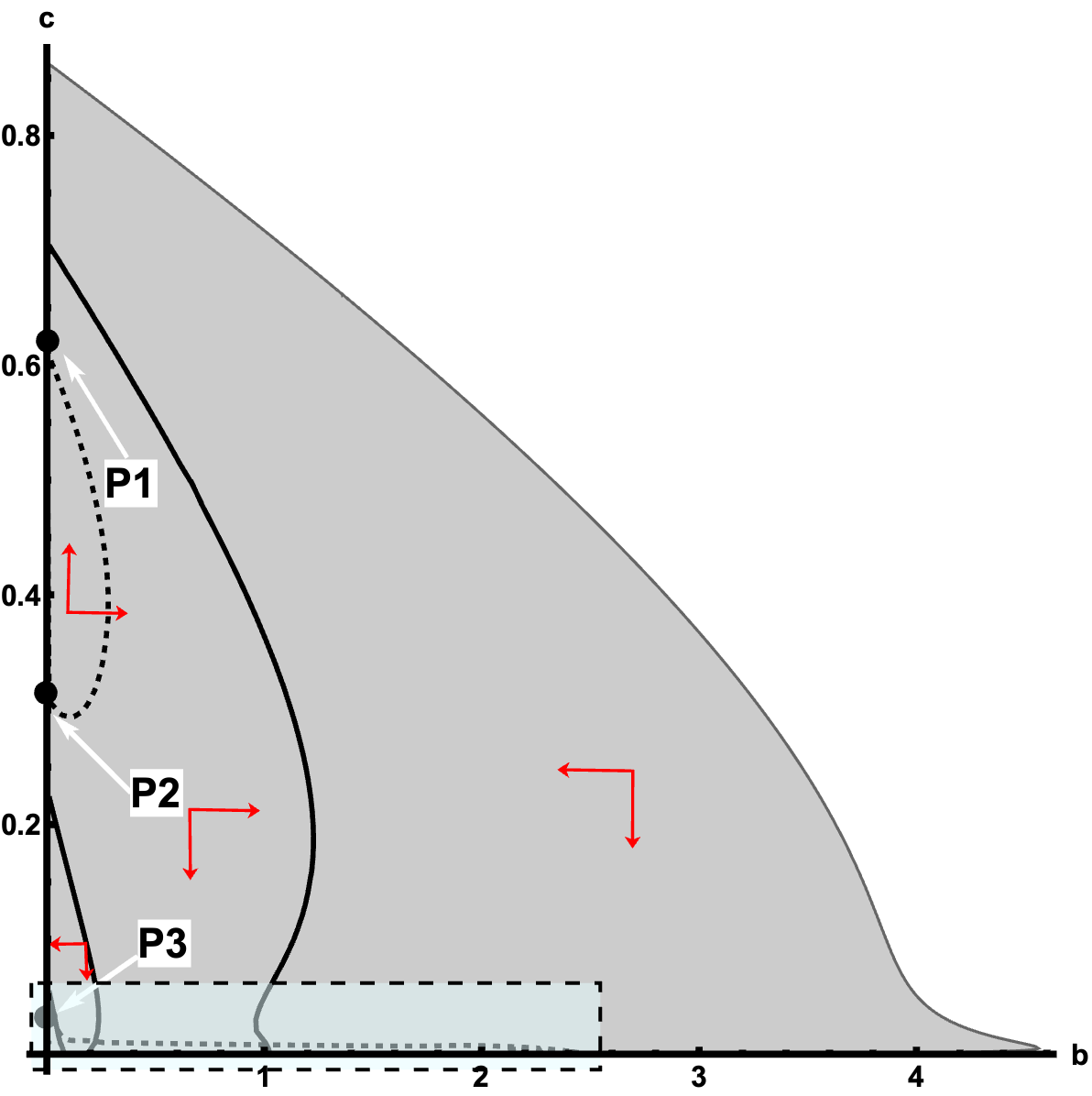}
\end{subfigure}
\begin{subfigure}[b]{.5\linewidth}
\includegraphics[width=6.5cm, height=6.5cm]{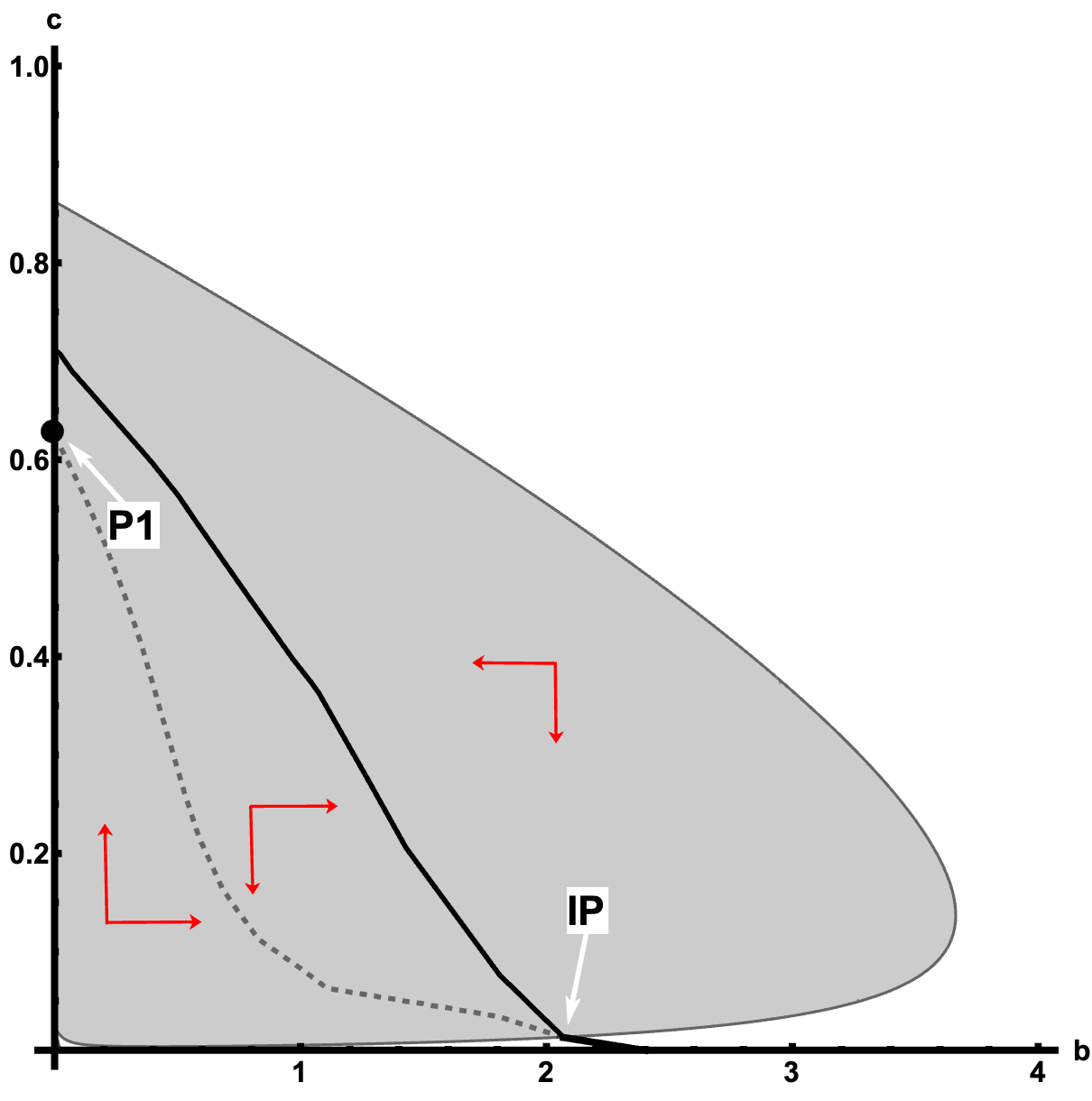}
\end{subfigure}
\begin{subfigure}[b]{.4\linewidth}
\includegraphics[width=6.5cm, height=6.5cm]{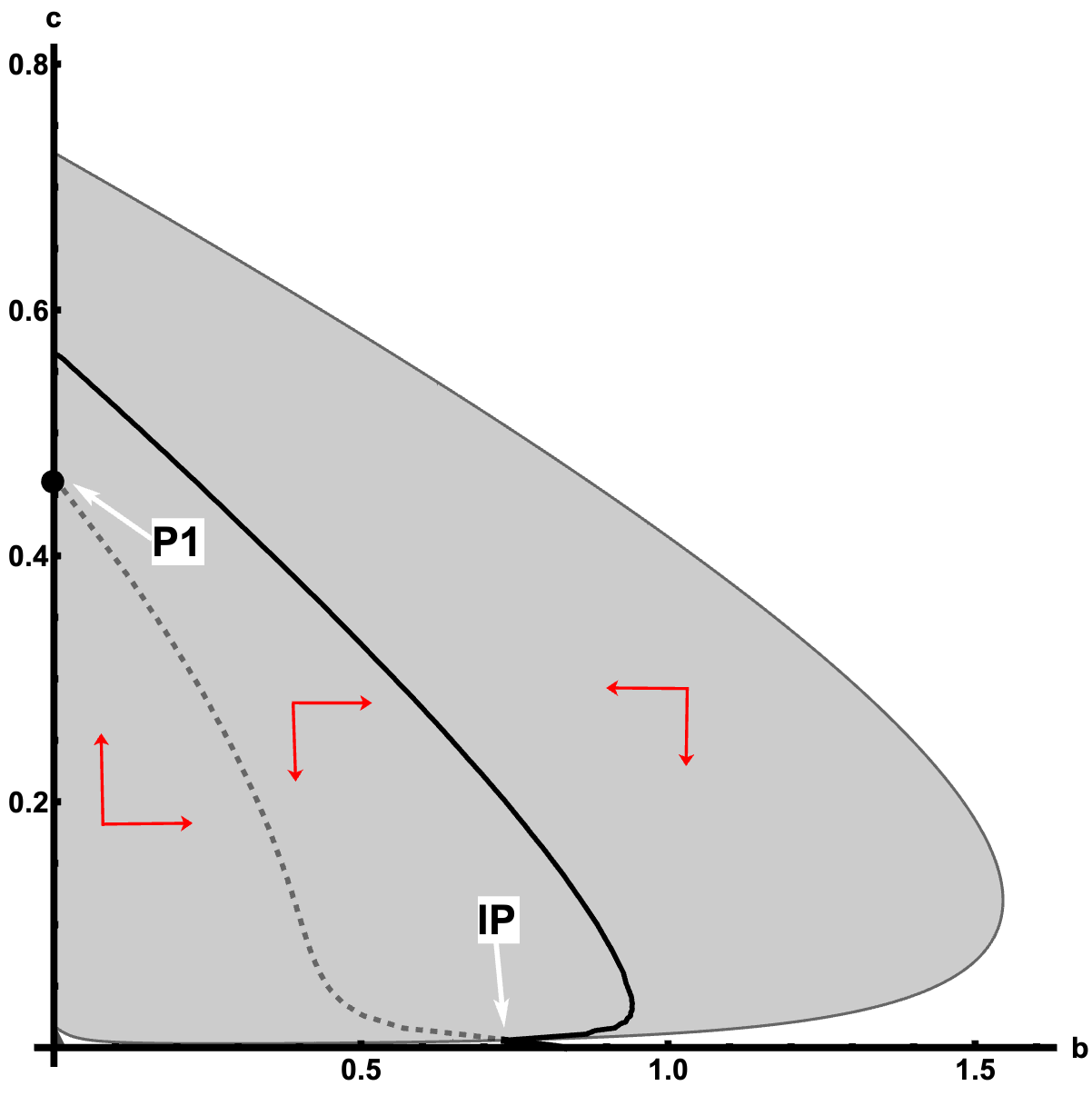}
\end{subfigure}
\caption{Evolutionary phase plane for the evolving parameter pair $(b,c)$. \emph{Top panel}: $\rho$ and $\gamma$ as given in (\ref{another1}) and (\ref{another2}) with parameter values in (\ref{setparam1}). \emph{Bottom panels}: $\rho$ and $\gamma$ as given in (\ref{another12}) and (\ref{another22}) and parameter values in (\ref{setparam2}) (\emph{left}) and (\ref{setparam3}) (\emph{right}). In light grey, the region of convergence to the stable limit cycle. P1, P2, P3: singularities on the $c$-axis. IP: intersection point of the isoclines on the HB-line. Continuous line: $b$-isocline. Dashed line: $c$-isocline. Red arrows: direction of evolution. Shaded area: numerically unreliable.}
\label{fig:sec43}        
\end{figure}

In Figure \ref{fig:sec44}, we finally give the contour-plot of the invasion fitness $s_{(c_1,0,c_2,0)}(b,c)$ in the $(b,c)$-plane to check if the evolutionary stable coexistence of two fixed strategies $c_1$ and $c_2$ and fixed $b=0$ can be invaded by a density dependent strategy $(b,c)$. As the fitness of a mutant in the resident environment settled by the dimorphic coalition is positive when $b$ evolves, we confirm that the strategy $(b,c)$ with positive invasion fitness can invade the dimorphic singularities for $\rho$ and $\gamma$ in (\ref{another12}) and (\ref{another22}) and both sets of parameter values in (\ref{setparam2}) and (\ref{setparam3}).

\begin{figure}[H]
\centering
\begin{subfigure}[b]{.5\linewidth}
\includegraphics[width=6cm, height=6cm]{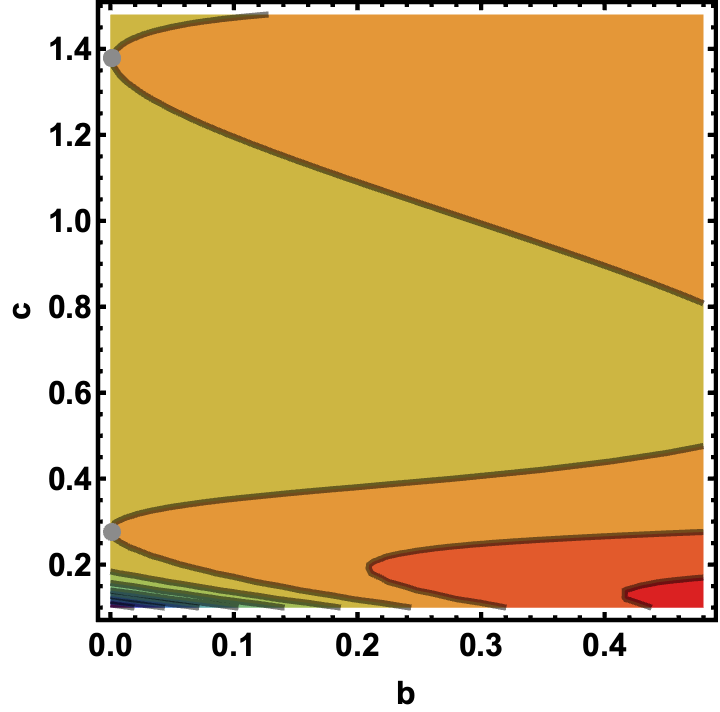}
\end{subfigure}
\begin{subfigure}[b]{.1\linewidth}
\includegraphics[width=1.5cm, height=6.4cm]{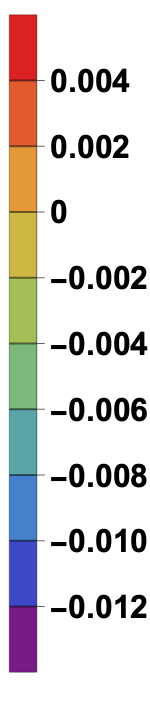}
\end{subfigure}
\begin{subfigure}[b]{.5\linewidth}
\includegraphics[width=6cm, height=6cm]{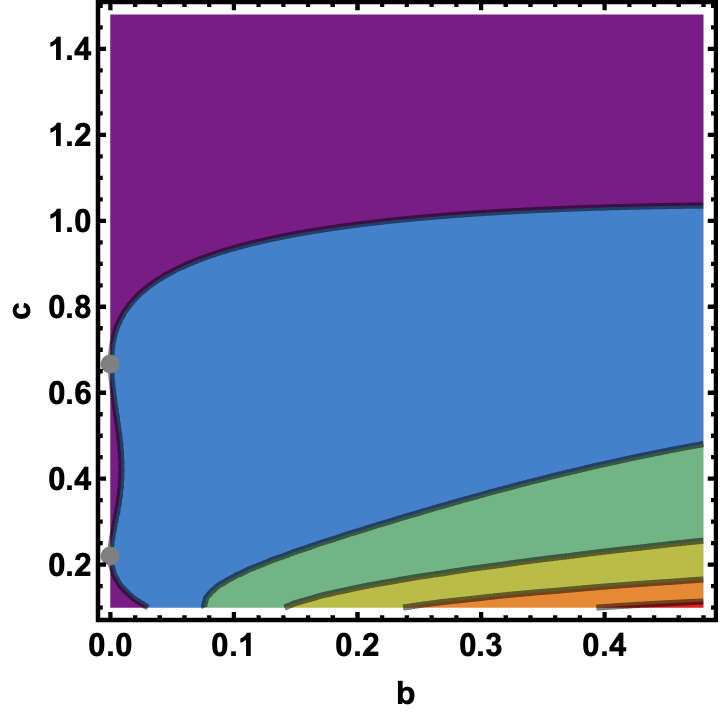}
\end{subfigure}
\begin{subfigure}[b]{.1\linewidth}
\includegraphics[width=1.5cm, height=6.4cm]{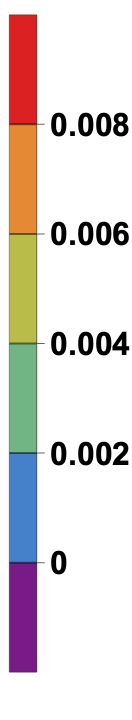}
\end{subfigure}
\caption{Can the evolutionarily stable coalition for evolving $c$ but fixed $b=0$ be invaded if $b$ is allowed to evolve as well? Contour-plots of the invasion fitness $s_{(c_1,0,c_2,0)}(b,c)$ in the $(b,c)$-plane and corresponding plotlegends. $\rho$ and $\gamma$ as given in (\ref{another12}) and (\ref{another22}) and parameter values in (\ref{setparam2}) (\emph{top panel}) and (\ref{setparam3}) (\emph{bottom panel}). Grey points: dimorphic coalition for $b=0$.}
\label{fig:sec44}        
\end{figure}

\section{Discussion}\label{sec6}
In this paper we use the adaptive dynamics framework to study the evolution of prey-density dependent handling times in a modified Rosenzweig-MacArthur predator-prey model (\cite{rosenzweig1963graphical}). In the model the handling time has a density dependent component and a density independent component. \\

The biological mechanism underlying the density dependence of the handling time in our model is simple and does not require any advanced cognitive abilities of the predator: during the handling of the prey, the predator can be distracted by the proximity of a prey individual which, with a given probability, tempts the predator to attack. Such events happen more frequently as the prey becomes more common. The overall effect is that the predator's handling of captured prey on average decreases as the population density of prey increases. The inclination to give in to the temptation to attack corresponds to the strength of density dependence.\\

We derive the functional response of the predators which is a Holling type II functional response with density dependent handling time and give the interpretation of the conversion factor of caught prey into predator offspring in terms of the individual ingestion rate per single predator.\\

The mechanistic modelling approach appears fundamental during the derivation of the ecological dynamics and, in particular, for the formulation of the predator functional and numerical responses (see \cite{Berardo:2020aa} for more details on the method). We focus on the \textit{bottom-up} approach to obtain a conversion factor with realistic underlying individual level behaviour and we point out that, in many cases, a phenomenological approach can lead to population functions that are meaningless in terms of individual dynamics. \\

We suggest a one-to-one relation between the conversion factor and the function modelling prey ingestion which looks in line with the dynamics energy budget theory (\cite{van2006introduction}, \cite{sousa2008empirical}, \cite{kooijman2010dynamic}), where birth is energy limited. An alternative to our model would consider an ordinary differential equation for the energy intake according to the idea that the resources are allocated to different metabolic processes and apply time-scale separation between the slow prey-predator dynamics and the dynamics of the energy uptake linked to predation. However, for the purpose of this study, it is preferable and easier to skip these steps and assume that the solution of the equation for the resource energy is given by the function $\rho(\tau)$. \\

We have shown, by means of examples, that a flexible handling time which varies inversely with the prey population density is selectively superior to a fixed handling time in the sense that if it can evolve, then (i) evolutionary branching of fixed handling times as seen in the model of \cite{geritz2007evolutionary} does not happen anymore, and (ii) evolutionarily stable coalitions of fixed handling times as seen in the model of \cite{geritz2007evolutionary} are no longer evolutionarily stable because they can be invaded. \\

We found that as long as the population is cycling, the density dependent component of the handling time evolves in strength up to a point where the population undergoes a supercritical Hopf bifurcation, and the cycles are lost. From that point onwards evolution effectively stops, because without cycles, i.e., in a constant population, there exists a continuum of selectively neutral combinations of the density dependent and independent components where the prey population density is minimised. \\

Increasing the strength of the density dependent component of the handling time was found to have a stabilising effect on the population dynamics. The introduction of a prey-density dependent handling time, however, can lead to more complicated population dynamics than in the original Rosenzweig-MacArthur model. For example, we have seen multiple positive equilibria and the simultaneous existence of stable and unstable limit cycles, neither of these phenomena occurs in the standard model. Much of this remains underexplored because our primary interest was evolutionary, and in the examples that we studied evolution steered clear from areas in the evolutionary trait space where such complications occur.\\

A practical point: During our analysis we often had to deal with large limit cycles that came very close to the predator or prey axes. Aside from the question whether such cycles are biologically reasonable (because of the high risk of random extinction when the prey or predator is rare), they also pose a numerical problem. The problem arises because successive loops of an orbit converging to the limit cycle become extremely closely packed where the limit cycle almost touches one of the axes. Small numerical errors during this phase of the limit cycle causes big errors elsewhere, as a consequence of which integration of, e.g., fitness over the limit cycle becomes unreliable. We found that the problem is almost completely eliminated by using a non-linear scaling of the population densities that effectively stretches the space near the axes (but not elsewhere) so that the successive loops of an orbit converging to the limit cycle are no longer closely packed. For the details we refer to the \ref{app6}.

\appendix

\section{Mechanistic derivation of the functional response}\label{app1}

We consider the following individual fast time scale reactions for the predators divided between searching $S$ and handling $H$
\begin{eqnarray}\nonumber
& \mathcircled{S} &  \xrightarrow{\tilde{a}x} \mathcircled{H} \quad  \textit{the searching predator attacks the prey and starts handling}\\ \nonumber
& \mathcircled{H} &  \xrightarrow{\tilde{b}x} \mathcircled{S}\quad \textit{the predator quits handling with prey-dependent rate}\\ \label{indlevel}
& \mathcircled{H} &  \xrightarrow{\tilde{c}} \mathcircled{S} \quad \textit{the predator quits handling spontaneously}
\end{eqnarray}

The dynamical system for the interactions in (\ref{indlevel}) is given by

\begin{equation}\label{dyn1}
\left\{\begin{array}{l}
\frac{dx}{d\tau}=g(x)x-\tilde{a}x\tilde{S}\\
\frac{d\tilde{y}}{d\tau}=\gamma(h)f(x,h)\tilde{y}-\delta \tilde{y}\\
\frac{d\tilde{S}}{d\tau}=-\tilde{a}x\tilde{S}+(\tilde{b}x+\tilde{c})\tilde{H}-\delta \tilde{S}\\
\frac{d\tilde{H}}{d\tau}=\tilde{a}x\tilde{S}-(\tilde{b}x+\tilde{c})\tilde{H}-\delta \tilde{H}
\end{array}\right.
\end{equation}
where $g(x)$ denotes the \emph{per-capita} growth rate of the prey population if the predator is absent, $f(x,h)$ is the predator functional response and $\gamma(h)$ is the total number of predator offspring from a single catch as given in (\ref{biggammadef}). The parameter $\delta$ represents the death rate of the predators, that we assume to be the same for searching and handling predators.  
The total prey $x$ and total predators $y$ are constant on the fast time scale. Moreover we suppose that the predator population size is much smaller than the prey population $y\ll x$: in this way the total prey population is not affected by the capture on the fast time scale. We assume $\varepsilon>0$ a small dimensionless scaling parameter and use the following scalings for the rates of the fast time reactions and for the predator population: $\tilde{a}=\varepsilon^{-1}a$, $\tilde{b}=\varepsilon^{-1}b$, $\tilde{c}=\varepsilon^{-1}c$, $\tilde{y}=\varepsilon y$, $\tilde{H}=\varepsilon H$, $\tilde{S}=\varepsilon S$. The slow-fast equations corresponding to the system above are then given by

\begin{equation}\label{dyn2}
\left\{\begin{array}{l}
\frac{dx}{d\tau}=g(x)x-\varepsilon^{-1}a x\varepsilon S\\
\frac{d\varepsilon y}{d\tau}=\gamma(h)f(x,h)\varepsilon y-\delta \varepsilon y \\
\frac{d\varepsilon S}{d\tau}=-\varepsilon^{-1}ax\varepsilon S+(\varepsilon^{-1}bx+\varepsilon^{-1}c)\varepsilon H-\delta \varepsilon S\\
\frac{d\varepsilon H}{d\tau}=\varepsilon^{-1}ax\varepsilon S-(\varepsilon^{-1}bx+\varepsilon^{-1}c)\varepsilon H-\delta \varepsilon H
\end{array}\right.
\end{equation}

On the fast time scale with $t=\varepsilon^{-1}\tau$ and $\varepsilon\rightarrow 0$, the equations become

\begin{equation}\label{fastdyn}
\left\{\begin{array}{l}
\frac{dx}{dt}=0\\
\frac{dy}{dt}=0 \\
\frac{dS}{dt}=-axS+(bx+c)H\\
\frac{dH}{dt}=axS-(bx+c)H
\end{array}\right.
\end{equation}

From the equations for the fast dynamics $S$ and $H$, we can therefore derive the equilibrium for the fast variables $\hat{S}=\frac{(bx+c)y}{bx+c+ax}$ and $\hat{H}=\frac{axy}{bx+c+ax}$. The corresponding functional response is given by definition

\begin{equation}\label{funcresp}
f(x,h(x))=\frac{ax\hat{S}}{y}=\frac{ax(bx+c)}{bx+c+ax}=\frac{ax}{1+\frac{1}{bx+c}ax}.
\end{equation}
The function in (\ref{funcresp}) is a Holling type II like functional response with density dependent handling time $\frac{1}{bx+c}$.

\section{Scaling of the equations of the resident dynamics}\label{app4}
To reduce the number of parameters we rewrite the system in (\ref{monoone}) and (\ref{monotwo}) in terms of the dimensionless quantities $\tilde{x}=\frac{x}{K}$, $\tilde{y}=\frac{\beta y}{r}$, $\tilde{b}=\frac{b}{\beta}$, $\tilde{c}=\frac{c}{\beta K}$, $\tilde{\delta}= \frac{\delta}{r}$, $\tilde{t}=rt$, $\tilde{\tau}=\beta K\tau$ and the functions $\tilde{h}(\tilde{x})=\beta K h(K \tilde{x})$, $\tilde{\rho}(\tilde{\tau})=\frac{1}{r}\rho\left(\frac{\tilde{\tau}}{\beta K}\right)$, $\tilde{f}(\tilde{x},\tilde{h}(\tilde{x}))=\frac{1}{\beta K}f(K\tilde{x},\frac{\tilde{h}(K\tilde{x})}{\beta K})$, $\tilde{\gamma}(\tilde{h}(\tilde{x}))=\frac{\beta K}{r}\gamma(\frac{\tilde{h}(K\tilde{x})}{\beta K})$. Then we get

\begin{eqnarray}
\frac{d\tilde{x}}{d\tilde{t}}&=&\tilde{x} (1-\tilde{x})-\tilde{f}(\tilde{x},\tilde{h}(\tilde{x}))\tilde{y}\\
\frac{d\tilde{y}}{d\tilde{t}}&=&\tilde{\gamma}(\tilde{h}(\tilde{x})) \tilde{f}(\tilde{x},\tilde{h}(\tilde{x}))\tilde{y}-\tilde{\delta}\tilde{y}
\end{eqnarray}

with

\begin{eqnarray}
\tilde{h}(\tilde{x})&=&\frac{1}{\tilde{b}\tilde{x}+\tilde{c}}\\
\tilde{f}(\tilde{x},\tilde{h}(\tilde{x}))&=&\frac{\tilde{x}}{1+\tilde{x}\tilde{h}(\tilde{x})}\\
\tilde{\gamma}(\tilde{h}(\tilde{x}))&=&\int_0^\infty \tilde{\rho}(\tilde{\tau}) e^{-\frac{\tilde{\tau}}{\tilde{h}(\tilde{x})}} d\tilde{\tau}
\end{eqnarray}

Dropping all tildes, we effectively get the original system in (\ref{monoone}) and (\ref{monotwo}) wit $r=1$, $K=1$ and $\beta=1$.

\section{Sufficient but non-necessary result on the uniqueness of the interior equilibrium}\label{app2}

The interior equilibrium points are the real and positive intersections of the prey and predator isoclines. The prey isocline is obtained by imposing $\frac{dx}{dt}=0$ and presents the following formulation

\begin{equation}\label{preyiso}
y=\frac{(1-x) (b x+c+x)}{b x+c}.
\end{equation}
The prey isocline is positive for $x<1$ and has a unique positive root at $x=1$, while it intersects the $y$-axis at $y=1$. The function is increasing for $\frac{-b c-\sqrt{(b+1) c (b+c)}+c}{b (b+1)}\leq x\leq \frac{\sqrt{(b+1) c (b+c)}-(b+1) c}{b (b+1)}$ and decreasing for $x>\frac{\sqrt{(b+1) c (b+c)}-(b+1) c}{b (b+1)}$ with $\frac{\sqrt{(b+1) c (b+c)}-(b+1) c}{b (b+1)}<1$. Finally, the prey isocline is concave for every positive $x$.\\

The predator isocline is given by the positive root of the equation

\begin{equation}\label{prediso}
0=\gamma(h(x)) f(x,h(x)) -\delta 
\end{equation}

Given that $f(x,h(x))$ is positive and concave, monotonically increasing in $x$ and converging to an asymptote of the form $y=m x + n$ (with $m=\underset{x\to \infty }{\text{lim}}\frac{1}{1+\frac{ x}{\text{b} x+c}}$ and $n=\underset{x\to \infty }{\text{lim}}\frac{ x}{1+\frac{ x}{\text{b} x+c}}-m x$)  and $\gamma(h(x))$  is positive and monotonically decreasing in $x$ down to $0$, then their product is a positive function. We compute the value for $\gamma(h(x)) f(x,h(x))|_{x=0}=0$ and the limit of $\gamma(h(x))f(x,h(x))$ for $x\rightarrow \infty$ that is either $0$ or a constant. However, we do not have information on the monotonicity of the product of the conversion factor and the functional response. \\

In order to obtain existence of a unique interior equilibrium, we need the $x$-coordinate of the equilibrium point to be less than $1$ (the normalised carrying capacity) and the solution of the equation in (\ref{prediso}) to be unique. In many cases, the easiest way to check the number of roots of (\ref{prediso}) is through numerical methods. However, when $f(x,h(x))$ is given and $\gamma(h(x))$ allows it, it is possible to check the number of positive solutions for a certain class of functions $\gamma(h(x))$ by using the following lemma:

\begin{lem}\label{lemmadue}
Suppose $\log\left[ \gamma(h(x)) \right]'\neq - \log\left[ f(x,h(x)) \right]'$ for all $x\in (0,1)$; then, the equation $\gamma(h(x))f(x,h(x))=\delta$ has at most one root in $(0,1)$.
\end{lem}

\begin{proof}
Suppose that there are two roots. Then, by the mean value theorem, there exists $x_0 \in (0,1)$ such that $\left[\gamma(h(x))f(x,h(x))\right]'=0$ at $x=x_0$. Hence, $\log \left[\gamma(h(x)) \right]'= - \log \left[f(x,h(x)) \right]'$, that is a contradiction.
\end{proof}

The condition on the function $\gamma(h(x))$ for the uniqueness of the interior equilibrium in Lemma \ref{lemmadue} is sufficient but not necessary. Indeed, it possible to find $x_0$ such that $\log\left[ \gamma(h(x)) \right]'= - \log\left[ f(x,h(x)) \right]'$ for $x=x_0$, but only one positive solution of the equation $\gamma(h(x)) f(x,h(x))=\delta$ is in the range $(0,1)$ and then feasible. \\

Lemma \ref{lemmadue} suggests a graphical means to see at one glance whether a given function $\gamma(h(x))$ does not (or may) permit multiple positive equilibria for a given function $f(x)$. In Figure \ref{fig:appB1}, we have plotted $-\log[f(x,h(x))]$ with $f(x,h(x))=\frac{x}{1+\frac{1}{bx+c}x}$ and several vertical translations as functions of $x$. These lines are \emph{critical functions} for $\log[\gamma(h(x))]$ whose derivative must be nowhere the same as that of $-\log[f(x,h(x))]$ in order to satisfy the condition of the Lemma. However, if the condition in Lemma \ref{lemmadue} is violated, we cannot exclude uniqueness of the interior equilibrium. Counterexamples are the functions II and III which verify $\log\left[ \gamma(h(x)) \right]'= - \log\left[ f(x,h(x)) \right]'$ for respectively $x=0.09484$ and $x=0.25507$, but give only one solution in the range $(0,1)$ for the equation $\gamma(h(x))f(x,h(x))=\delta$.\\

When two predator isoclines (see Figure \ref{fig:sec31} in the main text) are present, the dynamics becomes more complicated and may, e.g., give rise to an Allee effect in the predator density (see \cite{freedman1986predator}, \cite{crawley1992population},  \cite{kot2001elements}, \cite{turchin2003complex}, \cite{zhu2003bifurcation}, \cite{bate2014disease}). The equilibrium corresponding to the second predator isocline is typically unstable, but the prey-only equilibrium becomes stable, leading to bi-stability with the stable limit cycle surrounding the coexistence equilibrium.\\

For the functions $\rho$ and $\gamma$ that we suggest in this article, the interior equilibrium is unique. In order to verify uniqueness of the interior equilibrium we use numerical methods. \\

\begin{figure}[H]
\centerline{
  \includegraphics[width=15cm,height=5cm]{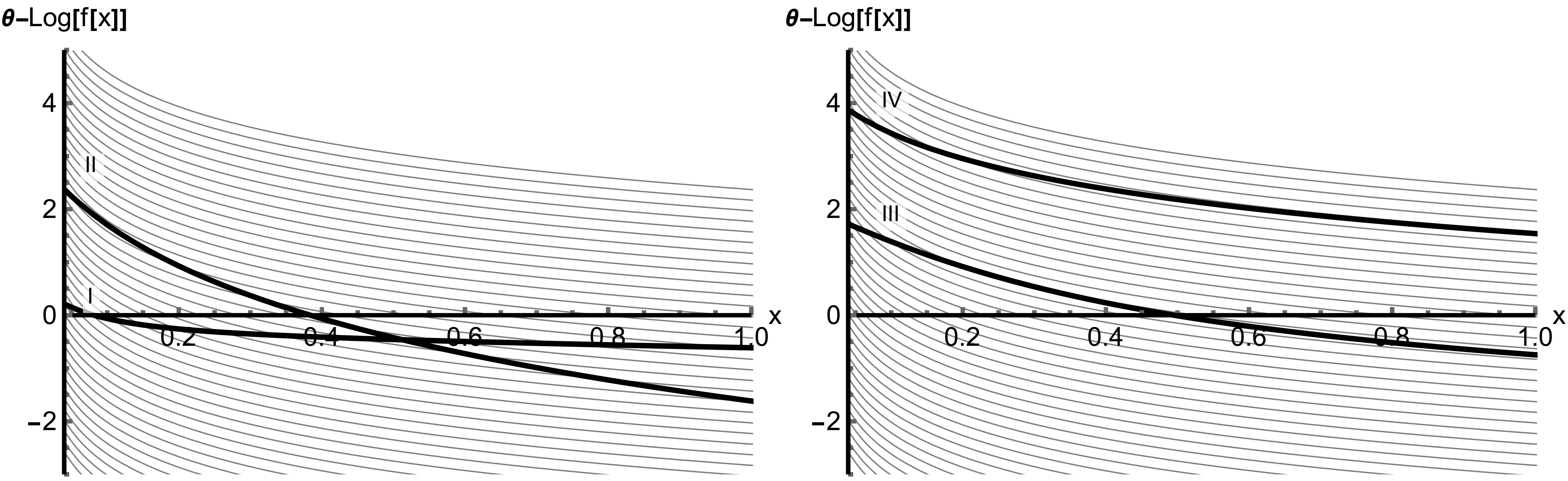}\hspace{0em}}
\caption{Plot of several vertical translations of $-\log\left[f(x,h(x))\right]$ (thin lines) with $\theta\in(-6,2)$ and different $\log\left[ \gamma(h(x)) \right]$ (thick lines). The values of $b$ and $c$ have been fixed at respectively $2$ and $0.1$.}
\label{fig:appB1}      
\end{figure}

\section{On the stability of the equilibria}\label{app5}
An equilibrium is asymptotically stable if the Jacobian matrix evaluated at the equilibrium has a positive determinant and a negative trace. If either or both have the opposite sign, then the equilibrium is unstable. The Jacobian matrix at any positive equilibrium is

\begin{equation}
J(x,y) = 
\begin{pmatrix}
f(x,h(x)) \left(\frac{g(x)}{f(x,h(x))}\right)' & -f(x,h(x))\\
y (\gamma(x) f(x,h(x)))' & 0
\end{pmatrix}
\end{equation}

 and so
 \begin{eqnarray}
 \det(J(x,y))&=& y f(x,h(x)) \left(\gamma(h(x)) f(x,h(x)) \right)' \\
 \text{tr}(J(x,y)) &=& f(x,h(x)) \left(\frac{g(x)}{f(x,h(x))}\right)'
 \end{eqnarray}
 
Thus, the sign of the determinant is given by the sign of $\left( \gamma(h(x)) f(x,h(x)) \right)'$, and the sign of the trace is given by the sign of $\left( \frac{g(x)}{f(x,h(x))} \right)$, both evaluated at the equilibrium prey density. Notice that $\left( \frac{g(x)}{f(x,h(x))} \right)$ is the slope of the $x$-isocline. Also notice that a potential lack of monotonicity of $\gamma(h(x)) f(x,h(x))$ not only the number of positive equilibria (see main text) but also their stability via the sign of the determinant of the Jacobian matrix.

\section{On the change of variables to approximate large cycles}\label{app6}
In this section we give the scaling of variables which results to be fundamental to approximate the large stable cycles for values of the parameters $b$ and $c$ close to the subcritical Hopf bifurcation and to detect the stable and unstable cycles in case of fold bifurcation for small values of $c$. We introduce the variables $u$ and $v$ such that

\begin{equation}\label{varchange}
x=u^m, \quad y=v^n
\end{equation}
with $m,n>1$ (typically in our simulations a convenient choice is $m=32$ and $n=16$). Given the definitions for the prey derivative $\frac{dx}{dt}$ and the predator derivative $\frac{dy}{dt}$ in (\ref{monoone}) and (\ref{monotwo}) , we derive the system of equations for $\frac{du}{dt}$ and $\frac{dv}{dt}$ 

\begin{eqnarray}\label{equ}
\frac{du}{dt}&=& \frac{1}{m} x^{\frac{1}{m}-1} \frac{dx}{dt} \Big|_{x=u^m, y=v^n},\\\label{eqv}
\frac{dv}{dt}&=& \frac{1}{n} y^{\frac{1}{n}-1} \frac{dy}{dt} \Big|_{x=u^m, y=v^n},
\end{eqnarray}
We solve the system in (\ref{equ}) and (\ref{eqv}) with initial conditions $u_0=\sqrt[m]{x_0}$ and $v_0=\sqrt[n]{y_0}$. The solution for $x$ and $y$ is obtained by applying back the equalities in  (\ref{varchange}).\\

In Figure \ref{fig:appE1}, we apply the change of variables to investigate the fold bifurcation and subcritical Hopf bifurcation in Figure \ref{fig:sec32}a-c of Section \ref{sec3}. In particular, we plot a \emph{forward orbit} (black) and a \emph{backward orbit} (red) starting with the same initial conditions but differing in the sign of the gradients $\frac{1}{m} x^{\frac{1}{m}-1} \frac{dx}{dt} \Big|_{x=u^m, y=v^n}$ and $\frac{1}{n} y^{\frac{1}{n}-1} \frac{dy}{dt} \Big|_{x=u^m, y=v^n}$ for $u$ and $v$, respectively. 
In the \emph{top panels}, starting for a small value of $c$, the forward orbit converges to the interior equilibrium and the backward orbit loops one time around the equilibrium in a clock-wise direction and then disappears to $x\rightarrow\infty$. In the \emph{middle panels}, the system has just crossed a fold bifurcation of a stable and an unstable cycle, while in the \emph{bottom panels}, the unstable cycle shrinks on the interior equilibrium which becomes unstable via the subcritical Hopf bifurcation.

\newpage

\begin{figure}[H]
\centering
\begin{subfigure}[b]{0.7\textwidth}
\includegraphics[width=\textwidth]{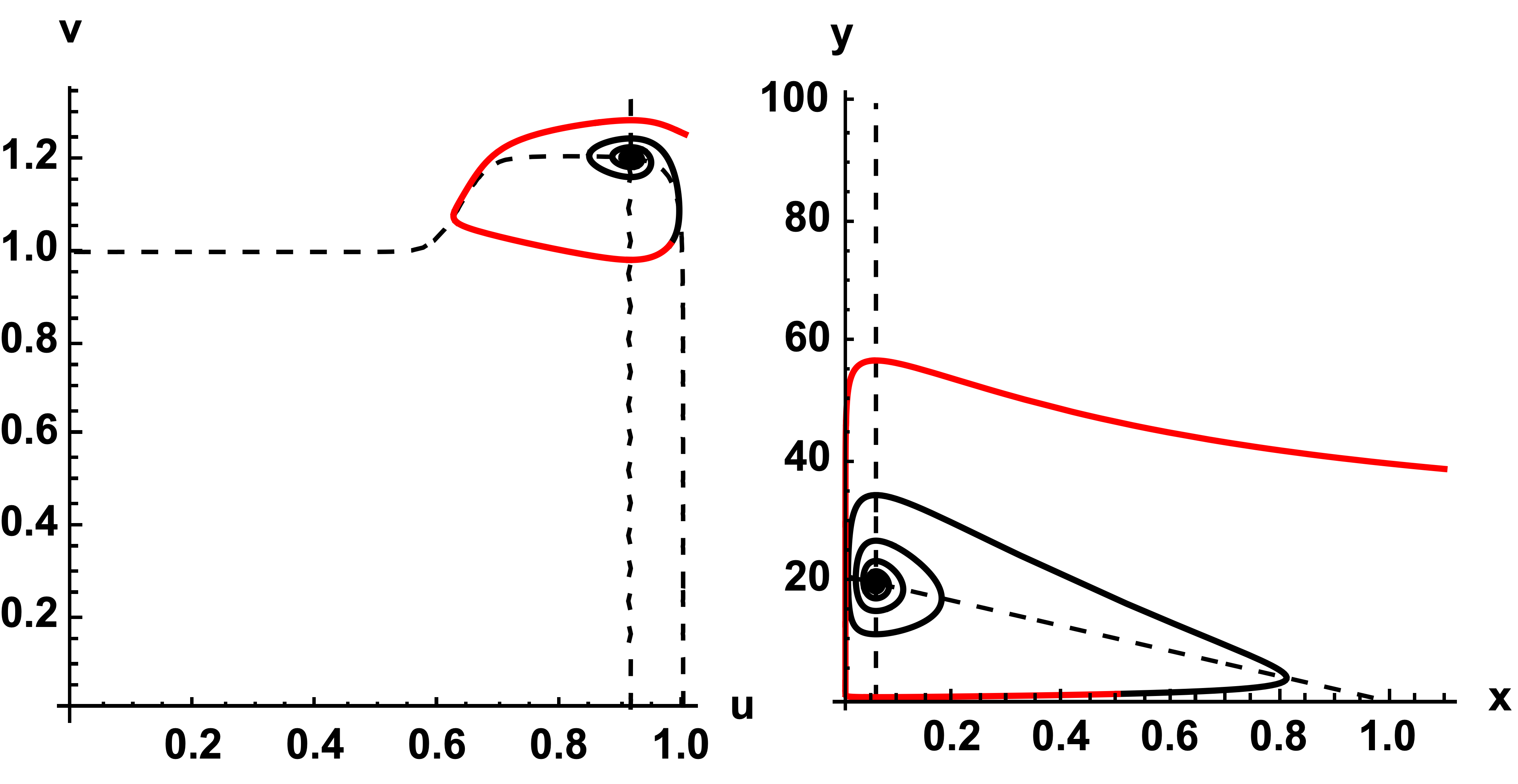}\hspace{0em}
\end{subfigure}
\begin{subfigure}[b]{0.7\textwidth}
\includegraphics[width=\textwidth]{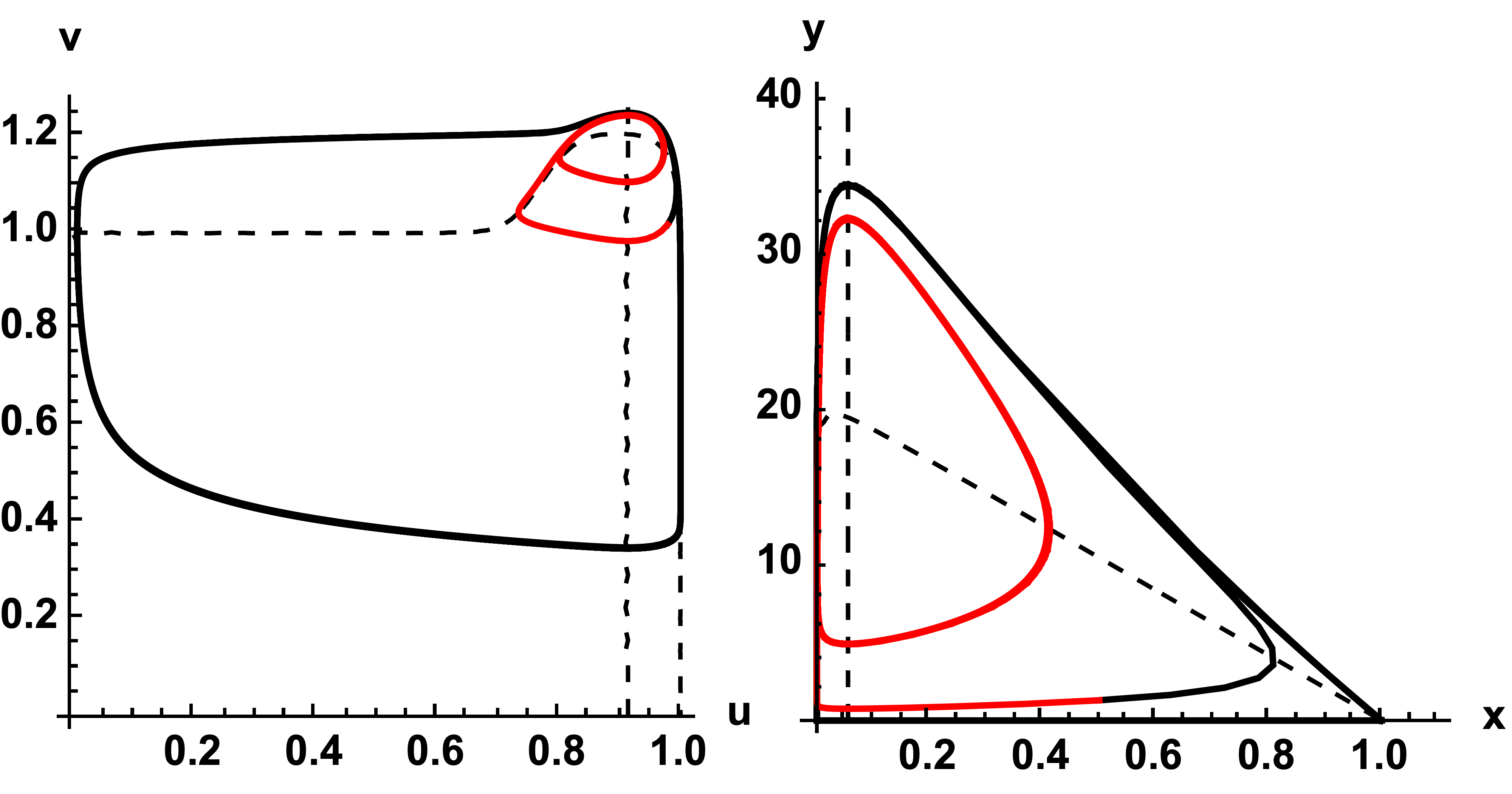}\hspace{0em}
\end{subfigure}
\begin{subfigure}[b]{0.7\textwidth}
\includegraphics[width=\textwidth]{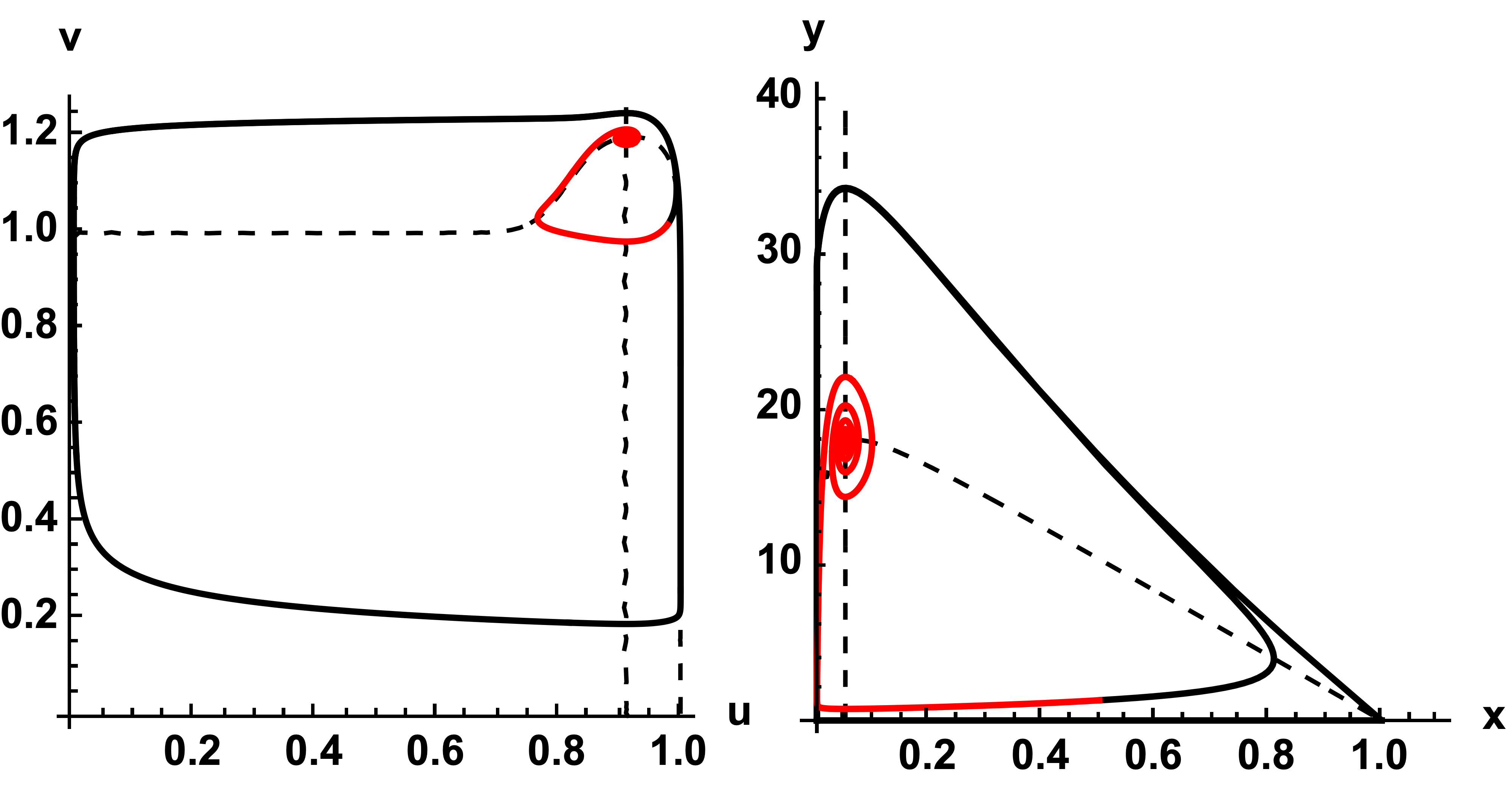}\hspace{0em}
\end{subfigure}
\caption{Plot of the solution of equations (\ref{equ}) and (\ref{eqv}) for the variable $v$ and corresponding $y$ with $\rho$ and $\gamma$ as given in (\ref{another1}) and (\ref{another2}) and parameter values: $m=32$, $n=16$, $\gamma_0=13.5651$, $\gamma_1=0.10002$, $\tau_0=0.11712$, $\tau_1=20$, $b=0.05$. \emph{Top panels}: $c=0.0000001$. \emph{Middle panels}: $c=0.00005$. \emph{Bottom panels}: $c=0.0003$. Initial conditions: $x_0=0.5$, $y_0=1.3$.  Dashed lines: system isoclines. Red lines: backward orbit. Black lines: forward orbit.}
\label{fig:appE1}        
\end{figure}

\section{On the optimisation principle for the prey density in a constant environment}\label{app3}
\subsection{Constant conversion factor}
We assume that the resident dynamics is settled on the interior equilibrium and the conversion factor is constant $\gamma(h)=\gamma_0$. We recall the invasion fitness for a mutant strategy $(b_m,c_m)$

\begin{equation}\label{futon}
s_{(b,c)}(b_{m},c_{m})=\gamma_0 f(x,h_{b_m,c_m}(x))-\delta.
\end{equation}

We assume the functional response being a Holling type II functional response with density dependent handling time $h_{b_m,c_m}(x)$ and therefore an increasing function of the prey density. Thus, the function $s_{(b,c)}(b_m,c_m)$ is uniformly-monotonous increasing in $x$ for all $(b_m,c_m)$. It follows that given a mutant strategy $(b_m,c_m)$ such that $x_{b_m,c_m}<x_{b,c}$ the invasion fitness verifies $0=s_{(b_m,c_m)}(b_m,c_m)<s_{(b,c)}(b_m,c_m)$, and viceversa when $x_{b_m,c_m}>x_{b,c}$, $0=s_{(b_m,c_m)}(b_m,c_m)>s_{(b,c)}(b_m,c_m)$. We conclude that the species with strategy $(b_m,c_m)$ can invade the species with strategy $(b,c)$ if and only if $x_{b_m,c_m}<x_{b,c}$. A pessimisation principle is present and evolution leads to minimisation of the prey density.

\subsection{Non decreasing conversion factor}
When the conversion factor $\gamma(h)$ is a non-decreasing function of the handling time, monotonicity of the fitness function must be checked. The invasion fitness for a mutant strategy $(b_m,c_m)$ is given by

\begin{equation}\label{fittwo}
s_{(b,c)}(b_m,c_m)= \gamma(h_{b_m,c_m}(x))f(x,h_{b_m,c_m}(x))-\delta.
\end{equation}
We use the one-to-one relation between $x$ and $h_{b,c}(x)=\frac{1}{bx+c}$ and the substitution $x=\frac{1-ch_{b,c}}{b h_{b,c}}$ in (\ref{fittwo}):

\begin{equation}
s_{(b,c)}(b_m,c_m)=\gamma(h_{b_m,c_m}) f\left(\frac{1-ch_{b_m,c_m}}{b h_{b_m,c_m}}, h_{b_m,c_m} \right).
\end{equation}
Note that if $s_{(b,c)}(b_m,c_m)$ is monotonic in $h_{b,c}$, so does in $x$. \\

While $\gamma(h_{b,c})$ is non-decreasing in $h_{b,c}$ the functional response $f(h_{b,c})=f\left(\frac{1-ch_{b,c}}{b h_{b,c}}, h_{b,c} \right)$ is decreasing in $h_{b,c}$. Therefore, we cannot exclude non-monotonicity and, if that is the case, then $\frac{ds(h)}{dh}=0$ for some value of $h$ (with $s(h)$ being a short notation for $s_{(b,c)}(b_m,c_m)$). Note that the trait $(b_m,c_m)$ is absorbed in the definition of $h$, such that if $s(h)$ is monotonic in $h$, then it is uniformly monotonic (i.e. for every parameter pair $(b_m,c_m)$).\\

It follows that we can construct critical functions $\gamma_k(h)$ with $k\in (0,\infty)$ which verify $\frac{ds(h)}{dh}=0$ when $s(h)$ is non-monotonic. In particular, given the equation

\begin{equation}\label{critgammaeq}
\frac{ds(h)}{dh}=\gamma'(h) f(x,h(x))+\gamma(h) f'(x,h(x))=0,
\end{equation}
we obtain the critical functions

\begin{equation}
\gamma_k(h)=\frac{k}{f(x,h(x))}=\frac{kh (c_mh-1-b_m)}{c_mh-1}.
\end{equation}
The functions $\gamma_k(h)$ take value zero for $h=0$ and increase to the asymptote $h=\frac{1}{c}$. \\

A sufficient condition for monotonicity of the function $\gamma(h)$ is that it must be nowhere tangent to the critical functions $\gamma_k(h)$. On the other hand, sufficient and necessary condition for non-monotonicity is that there exist $h\in(0,\frac{1}{c})$ and a value for $k>0$ such that $\gamma(h)$ is tangent to $\gamma_k(h)$, that is $\gamma(h)=\gamma_k(h)$ and $\gamma'(h)=\gamma_k'(h)$. Moreover, the conversion factor must be more convex or more concave than the critical function in a neighbourhood of $h$ and $\frac{ds}{dh}$ must change sign at $h$, i.e. there exist $\varepsilon>0$ such that $\gamma''(h_0)$ does not change sign for every $h_0\in \left(h_{b,c}-\varepsilon, h_{b,c}+\varepsilon \right)$. \\

The sufficient and necessary conditions for non-monotonicity are summarised in Lemma \ref{lemmauno}. In Figure \ref{fig:appC1}, we give examples with the conversion functions that we have discussed in Section \ref{sec2}. For some fixed values of $(b_m,c_m)$, both the functions $\gamma(h)$ are nowhere tangent to the critical functions and we can exclude non-monotonicity. The same result can be checked with the density plot of the prey size $x$ in the two parameter bifurcation plot. In Figure \ref{fig:appC2}, we study three different scenarios where the conversion factor is tangent to a critical function: the functions in the left and middle panels verify the conditions of Lemma \ref{lemmauno}. This is not the case for the function in the left panel, where the condition on the second derivative of $\gamma(h)$ does not apply (condition \ref{lemcondtre} in Lemma \ref{lemmauno}).

\begin{figure}[H]
\centerline{
  \includegraphics[width=15cm,height=6cm]{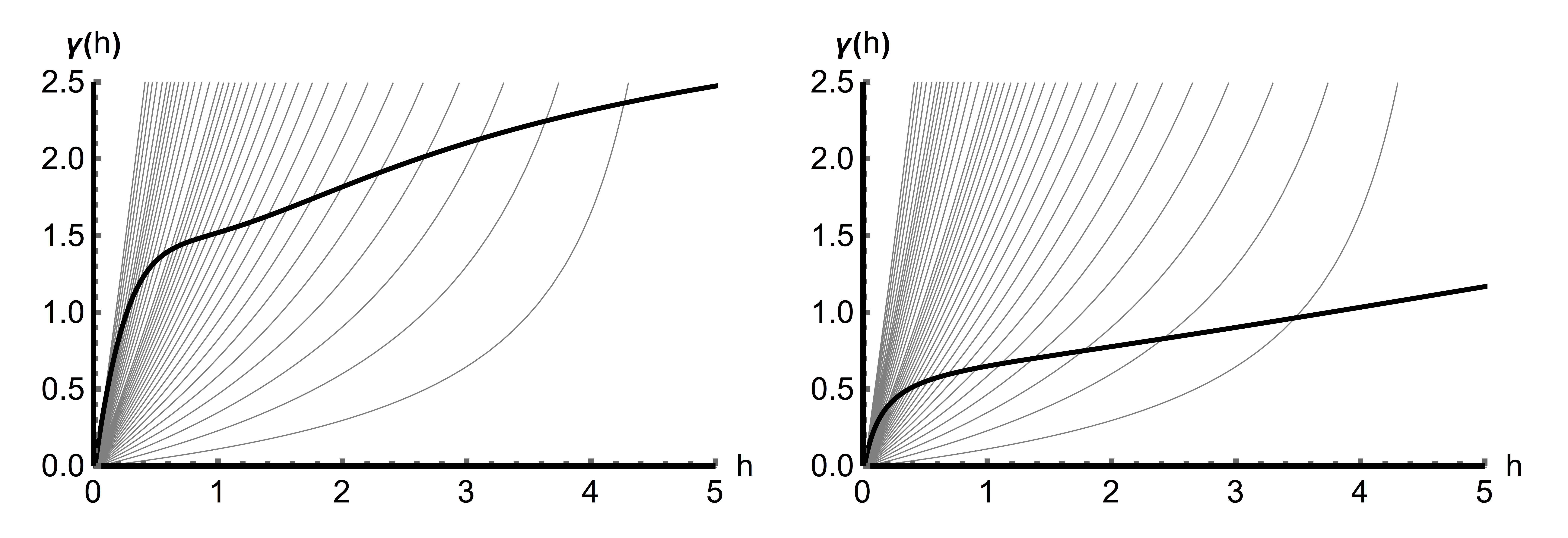}\hspace{0em}}
\caption{The critical functions $\gamma_k(h)$ with $k\in(0,1.2)$, $b_m=4$, $c_m=0.2$. \textit{Left panel:} the conversion function in (\ref{gamma2007}). \textit{Right panel:} the conversion function in (\ref{gamma2021}).}
\label{fig:appC1}      
\end{figure}

\begin{figure}[H]
\centering
\centerline{
  \includegraphics[width=14cm,height=4cm]{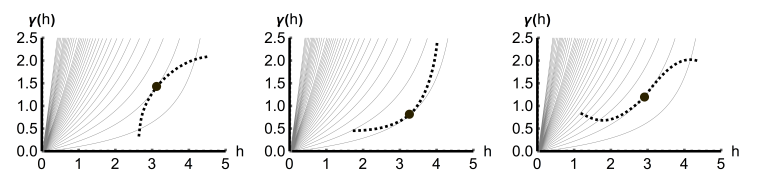}\hspace{0em}}
\caption{The critical functions $\gamma_k(h)$ with $k\in(0,1.2)$, $b_m=4$, $c_m=0.2$. Three possible cases when $\gamma(h)=\gamma_k(h)$ and $\gamma'(h)=\gamma_k'(h)$ (the function $\gamma(h)$ is sketched by the black dotted line). The \textit{left and middle panels} show examples of conversion factors such that $s(h)$ in non-monotonic. \textit{In the right panel}, $\frac{ds}{dh}$ does not change sign and therefore $s(h)$ is monotonic.}
\label{fig:appC2}      
\end{figure}

\bibliographystyle{abbrv}  
\bibliography{myad}
\nocite{*}

\thanks{
{\bf Acknowledgements}\\
This research was funded by the Academy of Finland, Centre of Excellence in Analysis and Dynamics Research.}

%
%
%

\end{document}